\documentclass[review, authoryear]{article}
\usepackage{graphicx}  
\usepackage{dcolumn}   
\usepackage{bm}        
\usepackage{amssymb}   
\usepackage{amsmath}   
\usepackage{mathrsfs}
\usepackage[normalem]{ulem}
\usepackage{color}
\usepackage{empheq}
\usepackage{hyperref}
\usepackage[modulo]{lineno}
\usepackage{float}
\usepackage{soul}
\usepackage{authblk}
\usepackage{amsthm}
\newtheorem{theorem}{Theorem}
\newtheorem{proposition}{Proposition}
\newtheorem{remark}{Remark}
\newtheorem{lemma}{Lemma}
\newtheorem{example}{Example}
\newtheorem{definition}{Definition}

\newcommand{\id}{\mathbb{I}}

\newcommand{\tr}{\mathrm{tr}}
\newcommand{\diag}{\mathrm{diag}}

\newcommand{\avg}{\mathbb{E}}

\newcommand{\cov}{\mathbb{C}}

\newcommand{\gauss}{\mathcal{N}}

\newcommand{\util}{\mathcal{U}}

\newcommand{\bq}{{\bm x}}
\newcommand{\bs}{{\bm s}}
\newcommand{\bw}{{\bm w}}
\newcommand{\bu}{{\bm y}}

\newcommand{\bv}{{\bm u}}
\newcommand{\bp}{{\bm p}}
\newcommand{\bpf}{\bm v}
\newcommand{\bmu}{{\bm \mu}}
\newcommand{\ba}{\bm \alpha}
\newcommand{\avgu}{\bm x_0}

\newcommand{\bsig}{\sqrt{\bm \sigma_0}}
\newcommand{\bom}{\sqrt{\bm \omega}}

\newcommand{\bO}{{\bm O}}
\newcommand{\bB}{{\bm B}}
\newcommand{\lam}{\bm \Lambda}

\newcommand{\siginf}{\bm\Sigma_0}
\newcommand{\hatsig}{\hat{\bm\Sigma}}
\newcommand{\sig}{\bm\Sigma}
\newcommand{\siginfleft}{ \mathcal{G}}
\newcommand{\siginfright}{ \mathcal{D}}
\newcommand{\om}{\bm\Omega}
\newcommand{\omleft}{ \mathcal{L}}
\newcommand{\omright}{ \mathcal{R}}
\newcommand{\hatomd}{\hat{\bm\Omega}^{\rm d}}
\newcommand{\omd}{{\bm\Omega}^{\rm d}}
\newcommand{\respinf}{\bm R_v}
\newcommand{\resp}{\bm R}
\newcommand{\respinfd}{{\bm R^{\rm d}_{v}}}
\newcommand{\hatrespinfd}{\hat{\bm R}^{\rm d}}
\newcommand{\respd}{{\bm R}^{\rm d}}
\newcommand{\hatrespd}{\hat{\bm R}^{\rm d}}
\newcommand{\comment}[1]{{{#1}}} 

\newcommand{\bV}{{\bm S}}
\newcommand{\bU}{\bm{W}}

\title{The Multivariate Kyle model: More is different}

\author[1]{L. C. Garc\'ia del Molino}
\author[1]{I. Mastromatteo}
\author[2,1]{M. Benzaquen}
\author[1]{J.-P. Bouchaud}
\affil[1]{Capital Fund Management, 23 rue de l'Universit\'e, 75007, Paris, France}
\affil[2]{Ladhyx and Department of Economics, UMR CNRS 7646, Ecole polytechnique, 91128 Palaiseau Cedex, France}
\begin{document}

\maketitle

\begin{abstract}
We reconsider the multivariate Kyle model in a risk-neutral setting with a single, perfectly informed rational insider and a rational competitive market maker, setting the price of $n$ correlated securities. We prove the unicity of a symmetric, positive definite solution for the impact matrix and provide insights on its interpretation. We explore its implications from the perspective of empirical market microstructure, and argue that it provides a sensible inference procedure to cure some pathologies encountered in recent attempts to calibrate cross-impact matrices. As an illustration, we determine the empirical cross impact matrix of US Treasuries, and compare the results with recent alternative calibration methods. 
\end{abstract}

\tableofcontents

\section{Introduction}

Understanding market impact -- the mechanism through which trades tend to push prices -- is with no doubt a venture of paramount importance. From the theoretical point of view, market impact lies at the very heart of price formation in financial markets. From the practitioners perspective, market impact is often at the origin of non negligible trading costs that need to be controlled to optimise  execution strategies.  In the past decades most of the literature has focused on the price impact of single products (for a recent review, see \cite{TQP}), with no regard of inter-asset interactions. However, many market participants trade large portfolios that combine hundreds or thousands of assets. Thus addressing the matter of inter-asset price impact, coined \emph{cross-impact}, is one of great interest both fundamentally and practically. Recent empirical studies show evidence of significant cross-impact effects in stock markets \cite{hasbrouck2001common,pasquariello2013strategic,wang2016cross,benzaquen2017dissecting}. A recurrent issue is that of empirical noise when it comes to large matrix estimation. It is thus important to continue to search for good statistical priors in order to help ``clean'' these large dimensional estimators.\\
 
What should one expect from theoretical economics on this matter? Which empirical observations should be considered unusual? Here we focus on how to make sense of the empirical observations accumulated over the past few years in the most orthodox setting: the classic Kyle model \cite{kyle1985continuous}, extended to a multi-asset framework. Our aim is to gain insight into (i) how information is diffused into prices (cross-sectionally) within the Kyle setting  and (ii) how one can use such results to regularise the very noisy regressions that arise in empirical cross-impact analysis. The multivariate Kyle model was first considered in \cite{caballe1994imperfect}, in a very general setting with $n$ assets and $m$ partially informed traders. The generality of the model is such that only a partial analysis of the solution is possible; furthermore, it involves an object that cannot be measured directly on empirical data, namely the correlation of the order flow of informed traders. Here we reconsider the problem with the issue of empirical calibration in mind. We focus on the particular case $m=1$ and provide a proof that the linear equilibrium is unique and characterized by a symmetric, positive definite (SPD) impact matrix. This solution provides an explicit recipe to infer cross impact from  (cleaned) order flow and return correlation matrices, that we compare to other natural recipes, such as Maximum Likelihood Estimators (MLE) \cite{benzaquen2017dissecting} or the recently proposed ``EigenLiquidity'' Model (ELM) \cite{mastromatteo2017trading}. \\

The paper is organised as follows. In sections~\ref{sec:problem} and \ref{sec:sol} we introduce the multivariate Kyle model and provide the equilibrium strategies, mostly building upon the work in \cite{caballe1994imperfect} but also providing new results. In section~\ref{sec:interpretation} we study in detail the mathematical properties of the solution. In section~\ref{sec:implications}, we introduce an impact estimator based on the equilibrium strategy of the market maker and compare it to other usual estimators both from the theoretical and empirical point of view. Along the text we present several examples intended to provide intuition behind the main results. In particular it is often interesting to confront the results of the multivariate and univariate models.
\\

In all the following bold uppercase symbols denote matrices, bold lower cases denote vectors and light lower cases denote scalars.

\section{The Multivariate Kyle Model}
\label{sec:problem}

In this section, we present in detail the multivariate Kyle setting and define the observables that would allow one to calibrate the model using empirical observables.  

\subsection{The Model}

Consider a single-period economy where three representative agents trade $n$ instruments. The agents are: an informed trader (IT) who has perfect information about the future prices $\bpf$, a noise trader (NT) that trades in absence of any information due to exogenous reasons, and a competitive market maker (MM) that has the role of enforcing price efficiency.
The dynamics of the model is set by the following rules.
\begin{enumerate}
\item \comment{A fundamental price $\bpf$ is sampled from a Gaussian distribution $\bpf \sim \gauss(\bp_0,\siginf)$ where $\siginf$ is SPD}. Only the IT knows the value of $\bpf$ in advance, while $\bp_0$ and $\siginf$ are common knowledge. \comment{We denote the price deviation from its mean as $\Delta \bpf := \bpf - \bp_0$.}
\item The IT and NT place simultaneously their orders of sizes $\bq$ an $\bv$ respectively. The bids of the NT $\bv$ are sampled from a Gaussian distribution $\bv \sim \gauss(0,\om)$ \comment{independent of the fundamental price}, where $\om$ is an invertible matrix.
\item The MM clears the excess demand $\bu$ at a clearing price $\bp$ based on the total observed order imbalance $\bu = \bq + \bv$, that allows him/her to form the best estimation of the fundamental prices. These fundamental prices are then revealed.
\end{enumerate}
The quantity $\bq$ requested by the {\it risk-neutral} IT is such that he maximises the expectation of his utility function:
\begin{equation}
 \util_{IT}(\bq,\bp) = \bq^\top  (\bpf - \bp) \, . \label{utilIT}
\end{equation}
Note that $\util_{IT}$ does not contain any risk penalty. While the introduction of such a penalty would affect some of the conclusions below, we decide to leave this interesting issue for future work. 
Consistent with the assumption that market making is competitive, the MM sets a price that matches in expectation $\bpf$ given the available public information, namely the total order imbalance $\bu$:
\begin{equation}
  \label{eq:MMcondition}
  \bp = \avg[\bpf|\bu]\, ,
\end{equation}
where $\avg[\cdot]$ denotes average with respect to the distribution of $\bpf$ and $\bv$.

Clearly, the above setting is highly stylized and, on many counts, unrealistic.\footnote{For example, the order flow in a Kyle setting has no temporal auto-correlations, whereas it is well known that the empirical order flow has long memory. See \cite{TQP} and \cite{benzaquen2017dissecting} for a recent discussion in the multivariate context.} Still, this model is able to capture some of the essential ingredients of a reasonable price-formation process: the information owned by ITs gets encoded into a trading order imbalance $\bq$ polluted by a noise $\bv$. In a competitive regime, MMs are expected to decode the information contained in the total order imbalance $\bu = \bq + \bv$, in order to provide the best possible prediction of the fundamental price $\bpf$. This mechanically induces market impact: because the order imbalance $\bu$ is correlated with the fundamental price, the traded price $\bp$ will also display a correlation with the order imbalance $\bu$, by that providing a sensible measure of market impact.

\subsection{Observables}

In order to gain insight on the implications of this model and in order to make testable predictions, one is required to provide some metrics that can be compared against market data. Luckily enough, the Gaussian nature of the setup allows one to only consider first- and second-order statistics (means and covariances) of prices and volumes in order to completely characterize the behavior of the model. Naturally, the fundamental parameters defining the model ($\om, \siginf$) are not directly observable, and need to be inferred from the statistics of trades prices $\bp$, and of volumes $\bu$, which are the only physical observables of the model.
\paragraph{Prices}
Due to the price efficiency condition (see Eq.~\eqref{eq:MMcondition}), the average traded price is equal to the \comment{average} of the fundamental price, itself equal (for consistency) to the initial price:
\begin{equation}
  \avg[\bp] = \avg[\bpf] = \bp_0 \, .
\end{equation}
However, the covariance of the traded price and that of the fundamental price have no reason to coincide:
\begin{align}
  \sig &= \cov[\bp,\bp] = \avg[(\bp - \bp_0)(\bp - \bp_0)^\top]\, ,\\
  \siginf &= \cov[\bpf,\bpf]= \avg[(\bpf - \bp_0)(\bpf - \bp_0)^\top] \, .
\end{align}
Further down we show that, at equilibrium, $\sig$ and $\siginf$ are exactly proportional (see Eq.~\eqref{eq:fund_vs_efficient}).
 
\paragraph{Volumes}
Due to the non-informed nature of the NT, the average order imbalance is fixed by the bias introduced by the IT, so that:
\begin{equation}
  \avg[\bu] = \avg[\bq] := \avgu \, .
\end{equation}
The relation between the portion of volume covariance due to the NT and the one due to the IT is more subtle, and is given by:
\begin{equation}
 \omd =  \cov[\bu,\bu] = \cov[\bq,\bq] + \om \, ,
\end{equation}
where we have introduced what we have coined the \emph{dressed volume covariance} $\omd$, which is the physically observable quantity. Correspondingly, the \emph{bare} volume covariance (i.e. not dressed by the noise contribution $\om$) is $\cov[x,x]$.

\paragraph{Response}
Ultimately, we are interested in characterising the expected price changes conditionally to a given trade imbalance: the response function, directly related to price impact. Within this model, the
physical observables measuring such quantities\comment{, that we coin \emph{dressed responses},} are:
\begin{align}
 \respd &= \avg[(\bp - \bp_0)\bu^\top] = \cov[\bp,\bu] \, , \\
 \respinfd &= \avg[(\bpf - \bp_0)\bu^\top] = \cov[\bpf,\bu] \, .
\end{align}
Due to the absence of correlations between the non-informed order imbalance $\bv$ and $\bpf$, one can relate the dressed responses with the \comment{responses with respect to $\bq$: $\resp$ and $\respinf$, that we coin \emph{bare responses},}  through
\comment{
\begin{align}
 \label{eq:dress_vs_bare_resp}
 \resp &= \avg[(\bp - \bp_0)\bq^\top]= \cov[\bp,\bq]= \respd - \cov[\bp,\bv] \, , \\
 \respinf &= \avg[(\bpf - \bp_0)\bq^\top]= \cov[\bpf,\bq] = \respinfd \, .
\end{align}}
Interestingly, while the fundamental prices are insensitive to the level of noise trading, the traded prices incorporate some degree of mispricing due to the spurious correlations
between the total order imbalance $\bu$ and the fundamental prices $\bpf$.


\section{Equilibrium strategies}\label{sec:sol}

In this section we characterise the equilibrium strategies of the multivariate Kyle model. Although most of the results can be inferred from the seminal work of Caball\'e and Krishnan \cite{caballe1994imperfect}, we believe it is useful to provide a streamlined version of our own proofs, which we present in a more pedagogical and in some cases more compact form and within which several theoretical issues and empirical implications can be discussed explicitly.

\subsection{Linear equilibrium}
Our main result is the existence of a unique linear equilibrium as defined below, whose structure can be expressed analytically.
\comment{\begin{definition}[Linear equilibrium] By \emph{linear equilibrium}, we mean a set of strategies in which the MM fixes the traded prices $\bp$ and the IT fixes their bid $\bq$ by means of linear rules
\begin{align}
  \label{eq:pricing_rule}
  \bp &= \bmu + \lam \bu \, ,\\
  \bq &= \ba + \bB \bpf \, ,
\end{align}
such that the two following conditions are satisfied:
\begin{enumerate}
 \item \emph{Profit maximization}: for all alternative strategies with $\bq'\neq\bq$,
\[\avg[\util_{\rm IT}(\bq,\bp)|\bpf]>\avg[\util_{\rm IT}(\bq',\bp)|\bpf].\]
 \item \emph{Price efficiency}: The price $\bp$ satisfies \eqref{eq:MMcondition}.
\end{enumerate}
\end{definition}}
Based on the assumption of a linear strategy for the MM as given in Eq.~\eqref{eq:pricing_rule}, we obtain the following results (proofs are provided in App.~\ref{app:solution}).
\begin{proposition}\label{th:IT_strategy}
Imposing the linear pricing rule~\eqref{eq:pricing_rule} for the MM implies that a rational IT will also set the order imbalance $\bq$ as a linear function of the imbalance, with:
\begin{equation}\label{eq:ITstrategy}
  \bq = \frac 1 2 \lam_{\rm S}^{-1} (\bpf - \bmu) \, ,
\end{equation}
where $\lam_{\rm S}$ denotes the symmetric part of $\lam$. Furthermore, the profit maximization condition for the IT's strategy implies that $\lam_{\rm S}$ has to be positive-definite (PD).
\end{proposition}


Eq.~\eqref{eq:pricing_rule} shows that $\lam$ plays the role of adjusting the traded price level $\bp$ proportionally to the imbalance $\bu$. Proposition \ref{th:IT_strategy} shows that $\lam$ also plays the role of setting the order imbalance from the informed trader $\bq$ given the knowledge of the fundamental price $\bpf$. Moreover, Eq.~\eqref{eq:ITstrategy} together with Eq.~\eqref{eq:pricing_rule} imply that both the order imbalance $\bu$ and the traded price $\bp$ are normally distributed random variables, due to stability of the Gaussian distribution under convolution. This last property is at the core of the following proposition.

\begin{proposition}
\label{prop:MM_strategy}
Assuming a linear strategy for the IT as in Eq.~\eqref{eq:ITstrategy} implies that the parameters of the pricing rule of Eq.~\eqref{eq:pricing_rule} for the MM are given by:
\begin{equation}\label{eq:MM_params}
\begin{aligned}
 \bmu &= \bp_0 - \lam \bu_0 \, , \\
\bm\Lambda & = \respinfd (\omd)^{-1}.
\end{aligned}
\end{equation}
\end{proposition}
Note that $\bu_0$, $\respinfd$ and $\omd$ \comment{depend on} $\bq$, and therefore of $\lam_{\rm S}$. Substituting \eqref{eq:MM_params} into Eq.~\eqref{eq:ITstrategy}  allows to close the system and find an equation for $\lam$. Note however that the as obtained system has many possible solutions. To constrain the latter we must impose the profit maximization condition introduced in Proposition \ref{th:IT_strategy}.

\begin{proposition}
\label{prop:symmetry}
Assume that there exists a solution to the utility maximization problem of the IT of the form given by Eq.~\eqref{eq:ITstrategy}, and to the pricing rule for the MM given by Eqs.~\eqref{eq:pricing_rule} and \eqref{eq:MM_params}. Then, the profit maximization condition that $\lam_{\rm S}$ has to be PD implies that $\lam$ is symmetric and satisfies the equation:
\begin{equation}\label{eq:core_lam_Sym}
 \frac14\siginf = \lam\om\lam\, .
\end{equation}

\end{proposition}

The symmetry of $\lam$ leads to a unique solution that  can be expressed in terms of the parameters of the problem (i.e.~$\siginf$ and $\om$). \comment{Denoting by $\sqrt{\bm Y}$ the unique (see Lemma~\ref{lemma:pos_def_sqrt} in  App.~\ref{app:sym_lin_eq})
  PD solution of the matrix equation $\bm X\bm X = \bm Y$ for a SPD $\bm Y$ we introduce the following theorem:}

\begin{theorem}[Existence and unicity of the linear equilibrium]
  \label{th:main_theorem}
  There exists a unique linear equilibrium given by the strategies~\eqref{eq:pricing_rule} and~\eqref{eq:ITstrategy} where:
  \begin{eqnarray}
    \label{eq:lambda}
    \lam &=&\textstyle{\frac 1 2} \omright^{-1} \sqrt{\omright \siginf \omleft} \omleft^{-1} \, ,  \\
    \bu_0 &=& 0 \, , \\
    \label{eq:mu}   \bmu &=& \bp_0\, .
  \end{eqnarray}
  Here $\omleft,\omright$ are a factorisation of the matrix $\om = \omleft \omright$ satisfying $\omleft = \omright^\top$.
\end{theorem}

Note that the existence result is a special case of Proposition 3.1 from Ref.~\cite{caballe1994imperfect}, specialised to the case of a single IT with perfect information. Nonetheless, the absence of other equilibria with non-symmetric $\lam$ lacked a published proof (although a similar discussion can be found in an unpublished paper by Caball\'e and Krishnan \cite{caballe1989insider}). Because the symmetry property of $\lam$ has crucial consequences on both price formation and pricing  (see discussion in Sections~\ref{sec:interpretation} and~\ref{sec:implications}) we believe that ruling out the existence of non-symmetric equilibria is an important step in the analysis of the multivariate Kyle model. This is the main novel theoretical contribution of the present paper.

\begin{example}[Solution of the univariate Kyle model]
  The specialisation of Theorem~\ref{th:main_theorem} to the $n=1$ case in which only one asset is traded yields the well-known solution derived in~\cite{kyle1985continuous} (light lower case symbols are scalar versions of the bold upper case and bold lower case ones):
  \begin{equation}
    \label{eq:lambda_1d}
    \lambda = \frac 1 2 \sqrt{\frac{\sigma_0}{\omega}} \, ,
  \end{equation}
indicating that the constant of proportionality between price and imbalances scales with the amount of price fluctuations $\sqrt{\sigma_0}$ and is inversely proportional to the typical fluctuations of the noise $\sqrt{\omega}$.
Intuitively, the larger price deviations the rational MM expects, the more weight he should give to volume imbalances in order to forecast the fundamental price to the imbalances.
On the other hand, the more noise there is in the system, the less the volume signal is reliable, so that the traded price is closer to the uninformed prior $p_0$.
\end{example}

It is important to note that Eq. \eqref{eq:core_lam_Sym} alone does not imply symmetry of $\lam$, as it is necessary to further impose PDness in order to obtain symmetry. In fact, there are symmetric solutions to \eqref{eq:core_lam_Sym} that are not PD. These solutions lead to efficient traded prices but they do not optimise the IT's utility.

\begin{example}[Saddle point solutions]
Let $\bq$ be a strategy of the form given in \eqref{eq:ITstrategy} and let $\bmu = \bp_0$ and let $\lam$ be a symmetric but not PD solution of \eqref{eq:core_lam_Sym}. By construction, the gradient of $\avg[ \util_{IT}]$ with respect to $\bq$ at this point is zero, satisfying the first order condition for the IT, but this does not guarantee that it is a maximum. We now take an arbitrary perturbation of the strategy $\bq + \delta\bq$ and compute the difference in utility between the two strategies: 
\begin{align*}
\nonumber\delta \avg[\util_{IT}] & =  \avg[\util_{IT}](\bq) - \avg[\util_{IT}](\bq + \delta\bq)\\
			 & = \avg[2\bq^\top\lam\delta\bq + \delta\bq^\top\lam\delta\bq - \delta\bq^\top(\bpf - \bp_0)]\\
  & =\delta\bq^\top\lam\delta\bq\, .
\end{align*}
One can see that $\delta \avg[\util_{IT}]$ is always positive for arbitrary $\delta \bq$ if and only if $\bm \Lambda$ is PD. Without the second order condition the IT could find {better} strategies than $\bq$ meaning that the solution that corresponds to that $\lam$ is not an equilibrium but a saddle point. In fact, when $\lam$ is not PD, $\avg[\util_{IT}]$ can be arbitrarily large.\\
\end{example}

In such a case, the IT would trade arbitrarily large quantities of some linear combinations of assets, and no equilibrium would be possible. 
Imposing a PD impact matrix $\lam$ allows the MM to fix the optimal strategy of the IT, which otherwise would be undetermined -- at least in the absence of a further risk penalty in the IT utility function.  

\subsection{Price and volume covariances of the equilibrium solution}
We now investigate what are the testable implications of the linear equilibrium described above. Here, we provide the relations between the observables that one would measure in such a situation.

\begin{proposition}[Camouflage]
  In the equilibrium, the expectation of the dressed order imbalance is equal to:
  \begin{equation}\label{eq:u0}
      \bu_0 = \avg[\bu] = \avg[\bq] = 0 \, .
  \end{equation}
  The IT order imbalance covariance can be expressed as:
  \begin{equation}\label{eq:Eqq}
  \cov[\bq, \bq] =\om \, ,
  \end{equation}
  implying that:
  \begin{equation}
    \label{eq:equil_omega}
    \omd = 2\om \, .
  \end{equation}
\end{proposition}
\begin{proof}
 The results in \eqref{eq:u0} and \eqref{eq:Eqq} are obtained by directly evaluating the expectations:
\begin{align*}
 \avg[\bq] &= \frac12\lam^{-1}\avg[\bpf - \bp_0] = 0 \, ,\\
 \cov[\bq, \bq] & = \frac14 \lam^{-1}\avg[(\bpf - \bp_0)(\bpf - \bp_0)^\top]\lam^{-1} = \frac14 \lam^{-1}\siginf\lam^{-1}=\om \, .
\end{align*}
 The last result follows trivially from \eqref{eq:Eqq} and the definition of $\omd$.
\end{proof}

This result implies that, in addition to matching the trading level of the NT in order to conceal his information (as obtained in the one-dimensional case), the IT is required to match the \emph{directions} in which the NT trades (i.e. the eigenvectors of $\om$) such that the trading behaviour of the IT is indistinguishable from that of the NT.
An even more interesting consequence of this behaviour is given in the next proposition.

\begin{proposition}[Information diffusion]
  The fundamental price covariance and the traded price covariance in the equilibrium are related by:
  \begin{equation}
    \label{eq:fund_vs_efficient}
    \sig = \frac 1 2 \siginf \, .
  \end{equation}
  \comment{Furthermore, the residual information is
  \begin{align}
  \nonumber \cov[\bpf,\bpf|\bp] &= \cov[\bpf,\bpf] - \cov[\bpf,\bp]\cov[\bp,\bp]^{-1}\cov[\bp,\bpf]\\
  &= \siginf - \bB\lam\siginf\sig^{-1}\siginf\lam^\top\bB^\top
  \end{align}
  which at equilibrium reads
  \begin{equation}
   \cov[\bpf,\bpf|\bp] = \frac12\siginf.
  \end{equation}}
\end{proposition}
\begin{proof}
 The traded price covariance is obtained as:
\begin{equation}
  \label{eq:relation_price_covariance}
  \sig = \cov[\bp, \bp] = \lam \omd \lam^\top = 2 \lam \om \lam = \frac 1 2 \siginf \, ,
\end{equation}
where the last equality comes from Eq.~\eqref{eq:lambdaS} in App.~\ref{app:sym_lin_eq}. To find the residual information on needs to compute the covariance between traded prices and fundamental prices:
\begin{equation*}
 \cov[\bpf,\bp] = \bB\lam\cov[\bpf,\bpf] = \bB\lam\siginf.
\end{equation*}
\end{proof}
While the fact that the two covariances are related is expected, the fact that they are proportional to one another
is non-trivial. Furthermore \eqref{eq:relation_price_covariance} implies that the role of the impact matrix $\lam$ is to ``rotate" the fluctuations of the volume in the direction of the fluctuations of the fundamental price.
The intuition behind this behaviour is that due to camouflage the IT will trade in the same direction as the NT.
Hence, a rational MM trying to enforce price efficiency will be required to convert fluctuations along the principal components of $\om$ into forecasts of price fluctuations along the principal components of $\siginf$, thus matching the expected covariance $\siginf$. The fact that the proportionality constant is equal to 1/2 (only half of the information is revealed) is of minor importance, and is a consequence of the single-period feature of the model. Extending the Kyle framework to multiple time steps would turn the constant factor 1/2 into a time-varying function expressing the rate at which information is incorporated into prices (see the original Kyle paper~\cite{kyle1985continuous} for the single asset case and~\cite{vitale2012risk,lasserre2004asymmetric} for respectively the discrete and continuous multi-asset case).\\

Finally, it is interesting to characterise the responses in the context of a multivariate Kyle model.

\begin{proposition}[Responses]
The propositions above and Eq.~\eqref{eq:dress_vs_bare_resp} imply that
\begin{equation}
 \label{eq:equil_resp}
 \respinf = \respinfd = \respd = 2 \resp  = \frac 1 2 \siginf \lam^{-1} \, ,
\end{equation}
so that the responses are uniquely determined by $\siginf$ and $\lam$.
\end{proposition}
\begin{proof}
 The equilibrium values of the responses are found by plugging the equilibrium strategies on their definitions:
\begin{align*}
 \respinfd = \respinf&= \avg[(\bpf - \bp_0)\bq^\top] = \frac12\avg[(\bpf - \bp_0)(\bpf - \bp_0)^\top]\lam^{-1} = \frac 1 2\siginf\lam^{-1} \, ,\\
  \respd & = \avg[(\bp - \bp_0)\bu^\top] = \lam\avg[\bu\bu^\top] =  2\lam\om = \frac 1 2\siginf\lam^{-1} \, ,\\
  \resp & = \avg[(\bp - \bp_0)\bq^\top] = \lam\avg[\bq\bq^\top] =  \lam\om=\frac 1 4\siginf\lam^{-1}\, .
\end{align*}
\end{proof}
An interesting consequence of this proposition is that within the present framework, the dressed responses should be entirely determined by $\siginf$ and $\om$ alone. In Section~\ref{sec:implications} we will see how this naturally leads to the calibration of an impact model in which part of the estimation noise can be reduced by imposing the response structure above as a reasonable prior.
Finally, note that while the dressed and bare responses of the fundamental price coincide,  the dressed response of the traded price is a factor 2 larger than the bare response, due to the spurious correlations between prices and volume transiently introduced by the NT.

\subsection{Utilities and competitive market making}
\label{sec:break_even}
In this model the utilities of the three agents are such that they sum to zero:
\begin{eqnarray}
     \util_{IT} &=& - \bq^\top  (\bp - \bpf) \nonumber \, , \\
     \util_{NT} &=& - \bv^\top  (\bp - \bpf) \nonumber \, , \\
     \util_{MM} &=& (\bq + \bv)^\top  (\bp - \bpf) \, .\nonumber
\end{eqnarray}
It is thus interesting to characterise how utilities are transferred from one agent to the other
 at equilibrium. 
\begin{proposition}[Utilities at equilibrium]\label{prop:utilities}
The utilities of the three agents at equilibrium are equal to
\begin{eqnarray}
     \avg[\util_{IT}] &=& \tr (\resp) \, ,  \\
     \avg[\util_{NT}] &=& - \tr (\resp) \, ,  \\
     \avg[\util_{MM}] &=& 0 \, .
\end{eqnarray}
\end{proposition}
\begin{proof}
 \begin{align*}
  \avg[\util_{IT}] &= - \avg[\bq^\top  (\bp - \bpf)] = \tr(\avg[ \bpf\bq^\top] -\avg[ \bp\bq^\top]) = \tr(\respinf - \resp) = \tr(\resp) \, ,\\
  \avg[\util_{NT}] &= - \avg[\bv^\top  (\bp - \bpf)] = \tr(\avg[ \bpf\bv^\top] -\avg[ \bp\bv^\top]) = \tr(0 - \resp) = -\tr(\resp) \, ,\\
  \avg[\util_{MM}] &= - \avg[\util_{IT}] - \avg[\util_{NT}] = 0\, .
 \end{align*}

\end{proof}

As in the standard Kyle model, the linear equilibrium of the Multivariate Kyle Model is such that wealth is transferred from the NT to the IT, which is able to capitalise his informational advantage. The role of the MM is more subtle:  we assumed that the MM enforces an efficient price, hence by construction he is not optimising his wealth. However, as we show in Proposition \ref{prop:utilities}, imposing price efficiency through the pricing rule \eqref{eq:MMcondition} results in a break-even for the MM in the sense that the expectation of his utility is 0. \comment{In the next example we show that the contrary is not true: imposing the break-even condition for the MM is in general not sufficient to ensure efficient prices.}

\begin{example}[Efficient prices and break-even condition]
Consider the break-even condition $\avg[\util_{MM}] = 0$. The expected utility of the MM in the univariate model reads:
\begin{equation*}
 \avg[\util_{MM}] = \frac12(\mu - p_0)^2\lambda^{-1} + \lambda\omega^{\rm d} - r_v^{\rm d} = 0.
\end{equation*}
Imposing $\mu = p_0$ rules out the presence of complex solutions regardless of $r_v^{\rm d}$ and $\omega^{\rm d}$ and therefore:
\begin{equation*}
\begin{aligned}
 \lambda &= r_v^d/\omega^d \, ,\\
\mu &= p_0\, .
\end{aligned}
\end{equation*}
Hence in one dimension a break-even condition for the MM, together with the requirement $\mu = p_0$ is equivalent to imposing price efficiency. This equivalence stems from the low dimensionality of the system because when we impose the break-even condition in one dimension we obtain one equation for one parameter ($\mu$ is fixed by imposing that the solution has to be real for any choice of parameters). \\ 

In the multivariate Kyle model this is not the case. In particular, the break-even condition writes
\begin{equation*}
 \avg[\util_{MM}] = \tr\left(\frac12(\bmu - \bp_0)(\bmu - \bp_0)^\top\lam_{\rm S}^{-1} + \lam\omd - \respinf\right)=0 \, . 
\end{equation*}
In this case we obtain again one equation but we now need to fix $\sim n^2$ parameters. This means that there are many ways to break-even while only one is efficient.
In a more restrictive setting where the MM has to break-even for {each asset individually}, the condition can be written as:
\begin{equation}\label{eq:diagonal}
{\rm diag}\left(\frac12(\bmu - \bp_0)(\bmu - \bp_0)^\top\lam_{\rm S}^{-1} + \lam\omd - \respinf\right) = 0. 
\end{equation}
In this case there are $n$ equations, which are still not enough to fix $\sim n^2$ parameters.\\
\end{example}
Therefore, in one dimension the condition $\avg[\util_{MM}] = 0$ is both a necessary and sufficient condition for the prices to be efficient. In multiple dimensions, although $\avg[\util_{MM}] = 0$ is still is a consequence of price efficiency, the converse is not true, as clarified in the next remark.


\begin{remark}[Sufficient and necessary conditions for efficiency]
\comment{The asset-wise break even condition \eqref{eq:diagonal} can be also written as:
\begin{equation}\label{eq:asset-wise}
\avg[ (p_i - v_{i})y_i  ] = 0\ , 
\end{equation}
for each asset $i$. Under the assumption $\bmu = \bp_0$ and in matrix form, \eqref{eq:asset-wise} leads to 
\begin{align*}
 0 &= \diag(\avg[(\bp - \bpf)\bu^\top])\\
   &= \diag(\lam\avg[\bu\bu^\top] - \avg[\bpf\bu^\top])\\
   &= \diag(\lam\omd - \respinf)
\end{align*}
which corresponds to the diagonal of \eqref{eq:MM_params}. In order to recover \eqref{eq:MM_params} (and not only its diagonal) we 
need the off-diagonal version of equation \eqref{eq:asset-wise}:
\[\avg[ (p_i - v_{i})y_j  ] = 0\ , \] 
for $j\neq i$.
Unfortunately, these terms are not directly related to any utility, and merely indicate that in order to
prices to be efficient, the MM should destroy any correlation between the order flow $y_j$ and the mispricing $p_i - v_{i}$.} 
\end{remark}
Hence, whereas in one dimension the behavior of a Kyle MM can be justified on the basis of utility and competitive market making arguments (see Ref.~\cite{kyle1985continuous}), it is much less clear why in the multi-asset case a competitive MM would enforce price efficiency on the basis of his/her own utility alone.

\section{Interpretation}
\label{sec:interpretation}
The properties of the linear equilibrium that have been illustrated in the section above provide
some intuition about the phenomenology of the model, but they do not help in making sense of Eq.~\eqref{eq:lambda} for the impact $\lam$. The following discussions are meant to provide a justification of the expression for $\lam$, so to shed some light on its structure. As far as we know, the following arguments have never appeared in the literature before.

\subsection{Whitening of prices and volumes}
The derivation of the linear equilibrium presented in Theorem~\ref{th:main_theorem} relies on the quadratic matrix equation \eqref{eq:core_lam_Sym} which  is solved by  (see proof of Proposition~\ref{prop:unique_sym}):
\begin{equation}
    \label{eq:change_var_lam}
    \lam = \frac 1 2 \siginfleft \bO \omleft^{-1} \, ,
\end{equation}
where we have used a decomposition for the price correlations $\siginf=\siginfleft\siginfright$ analogous to that of $\om= \omleft \omright$ and $\bO$ is an orthogonal matrix. The remarkable part of this finding is that, regardless of $\siginf$ and $\om$, there always exists a unique
change of basis $\bO$ such that $\lam$ is SPD.
Hence, the prediction process operated by the MM that results into the pricing rule $\Delta \bpf = \bpf - \bp_0= \lam \bu $ can be seen as resulting from the following steps:
\begin{enumerate}
    \item Apply the \emph{whitening} transformation $\omleft^{-1}$ to $\bu$, so to obtain the whitened imbalance $\tilde \bu = 2^{-1/2} \omleft^{-1} \bu$, with $\cov[\tilde \bu, \tilde \bu] = \id$.
    \item Apply the rotation $\bO$ to $\tilde \bu$, obtaining the whitened fundamental price variations $\Delta\tilde \bpf = \bO \tilde \bu$. As before, one has in fact $\cov[\Delta\tilde \bpf, \Delta\tilde \bpf] = \avg[\Delta\tilde \bpf \Delta\tilde \bpf^\top] = \id$.
    \item Apply the \emph{inverse whitening} transformation $\siginfleft$ in order to find the prediction for the fundamental price variation $\Delta \bpf = \bpf - \bp_0 = 2^{-1/2} \siginfleft \Delta\tilde \bpf $.
\end{enumerate}
While the first and the third steps of this procedure are intuitive, and can be seen as arising from dimensional analysis only, the nature of the rotation $\bO$ applied during the second step is less trivial, and will be investigated more closely in the following sections.

\begin{example}[Degeneracy in the one-dimensional Kyle model]
  In one dimension Eq.~\eqref{eq:core_lam_Sym}  becomes
  $$\frac{\sigma_0}{4} = \lambda^2 \omega\, , $$
  leading to the couple of solutions $\lambda = \pm \frac{1}{2} \sqrt{\sigma_0 / \omega}$. The degeneracy between the solutions is easily resolved by imposing positive-definiteness of the impact $\lam$, which selects the positive solution. In this case we have  $\siginfleft = \sqrt{\sigma_0}$, $\omleft = \sqrt{\omega}$, $\bO = 1$.
  Hence, the one-dimensional Kyle model has too low a dimensionality for anything non-trivial to happen from the point of view of $\bO$: it is only in higher dimension that one can appreciate its general structure.
\end{example}

\subsection{Basis of prices and basis of volumes}
As we show in the proof of Proposition \ref{prop:unique_sym}, the rotation $\bO$ that appears in Eq.~\eqref{eq:change_var_lam} can be expressed as
\begin{equation}
 \label{eq:rotation_simple}
 \bO = (\omright \siginfleft)^{-1} \sqrt{(\omright \siginfleft)(\omright \siginfleft)^\top}\, .
\end{equation}
In order to make some progress in understanding the complex nature of this rotation for $n>1$, it is important to notice that it is completely specified in terms of the matrix $\omright \siginfleft$.
Let us assume that the matrix factorization chosen to compute $\om$ and $\siginf$ is the Principal Component Decomposition:
\begin{eqnarray}
    \om & = & \underbrace{\bU \diag(\bom)}_{\omleft}
              \underbrace{\diag(\bom)\bU^\top}_{\omright} \, , \\
    \siginf & = & \underbrace{\bV \diag(\bsig) }_{\siginfleft}
                  \underbrace{\diag(\bsig) \bV^\top }_{\siginfright} \, ,
\end{eqnarray}
where $\bU$ and $\bV$ are orthogonal matrices, and where the vectors $\bom$ and $\bsig$ have positive elements.
Then we have:
\begin{equation}
    \omright \siginfleft = \diag(\bom)\bU^\top \bV \diag(\bsig) \, ,
\end{equation}
indicating that, besides the diagonal matrices $\diag(\bom)$ and $\diag(\bsig)$ that set the scale of the fluctuations of fundamental prices and order imbalancees, it is the overlap $\bU^\top\bV $ that links the eigenvectors of the volumes with that of the fundamental price. In order to obtain a stronger insight about the structure of $\bO$, we shall proceed with some simple examples.

\begin{example}[Two-dimensional Kyle model]\label{ex:2d}
  In two dimensions  the matrix $\bO$ can be characterised, without loss of generality, by a single angle $\theta$: \footnote{Indeed, by consistently choosing the factorisation of $\siginfleft$ and $\omleft$ one
  can fix arbitrarily the determinant of $\bO$ to be plus or minus one due to Eq.~\eqref{eq:rotation_simple}.}
\begin{equation*}
\bO = \begin{pmatrix}
\cos\theta&-\sin\theta\\
\sin\theta&\cos\theta
\end{pmatrix}.
\end{equation*}
Consider a system given by the following correlations:
\begin{equation*}
 \siginf = \begin{pmatrix}
            1 & \rho\\
	    \rho & 1
           \end{pmatrix},\qquad
 \om = \begin{pmatrix}
            \omega_{1} & 0\\
	    0 & \omega_{2}
           \end{pmatrix}.
\end{equation*}
In this case we can calculate $\bO$ explicitly. In particular one has
\begin{equation*}
\siginfleft= \frac{1}{\sqrt{2}} \begin{pmatrix}
            \sqrt{1 + \rho} & -\sqrt{1 - \rho}\\
	    \sqrt{1 + \rho} & \sqrt{1 - \rho}
           \end{pmatrix},\qquad
 \omleft^{-1} =\begin{pmatrix}
            \omega_{1}^{-1/2} & 0\\
	    0 & \omega_{2}^{-1/2}
           \end{pmatrix},
\end{equation*}
and imposing the symmetry of $\lam$ one can show that 
\[\theta = \arcsin\left(\frac{\omega_{1}^{-1/2}\sqrt{1 + \rho} + \omega_{2}^{-1/2}\sqrt{1 - \rho}}{\sqrt{2}\Delta}\right)\]
  where $\Delta = \sqrt{\omega_{1} + \omega_{2} + 2(\omega_{1}\omega_{2})^{1/2}\sqrt{1 - \rho^2}}$.
Finally we obtain
\begin{equation}\label{eq:2d}
\lam = \frac{1}{2\Delta}\begin{pmatrix}
1 + \sqrt{\frac{\omega_{2}}{\omega_{1}}}\sqrt{1-\rho^2}& \rho\\
\rho & 1 + \sqrt{\frac{\omega_{1}}{\omega_{2}}} \sqrt{1-\rho^2}
\end{pmatrix}.
\end{equation}
\end{example}
We can use the results in example \ref{ex:2d} to examine a couple of extreme cases: the one of extreme illiquidity and the one in which assets are strong correlations.

\begin{example}[Extreme illiquidity]\label{ex:extreme_illiquidity}    
  Let us consider the limit in which $\omega_{2} = \epsilon \omega$ with $\epsilon \to 0$, whereas $\omega_{1}=\omega$, implying that $y_2 \epsilon^{-1/2}$ is expected to be finite. Then the prediction of the MM \comment{up to first order in $\epsilon$} is:
  \begin{equation*}
  \bp - \bp_0 = \frac{1}{2 \sqrt{\omega}}
  \begin{pmatrix}
  y_1\\
  \rho y_1 + \sqrt{1-\rho^2}(y_2 \epsilon^{-1/2}) 
  \end{pmatrix} \, ,
\end{equation*}
implying that the efficient price of a liquid instrument is fixed solely by the volume traded on that same instrument, whereas for illiquid markets it is important to take into account quantities traded on liquid, correlated markets. On the other hand, the behavior of the IT \comment{up to first order in $\epsilon$} is given by:
\begin{equation*}
  \bq = \sqrt{\omega}
  \begin{pmatrix}
          \Delta v_1 \\
	  \sqrt{\frac{\epsilon}{1 - \rho^2}}(\Delta v_2 - \rho\Delta v_1)
         \end{pmatrix}\, .
  \end{equation*}
  Now the IT is encouraged to trade only the liquid asset and furthermore ignore the illiquid one when placing the bid. Similarly to the previous case, the traded price will remain of order one because the market orders for the second asset will be order $\sqrt{\epsilon}$.
\end{example}

\comment{Price innovation in example 6 can be interpreted as an indication of how an efficient market should work: the symmetry of cross-impact implies that the effect of trading one dollar of a very liquid asset $a_1$ on the price of an iliquid asset $a_2$ is the same as the effect of trading a dollar of asset $a_2$ on the price of asset $a_1$, however because there are going to be very few dollars traded on $a_2$ it will not affect significantly the price of $a_1$. Conversely, the price of $a_2$ will be heavily driven by the traded volume of $a_1$.}\\

Although we only proved the existence of the linear equilibrium for invertible $\siginf$, we can study the behavior of $\lam$ in the limiting case in which $\siginf$ has low rank as in the two following examples:

\begin{example}[Strongly-correlated prices]\label{ex:extreme_correlation}
  Consider $\rho = 1 - \delta$, with $\delta\to0$, so that $(\Delta p_1 - \Delta p_2) /\sqrt{\delta}$ is finite.
  According to \eqref{eq:2d} the prediction of the MM \comment{up to first order in $\delta$} is:
\begin{equation*}
  \bp - \bp_0 = \frac{1}{2 \sqrt{\omega_{1} + \omega_{2}}}
  \begin{pmatrix}
    y_1 + y_2\\
    y_1 + y_2 
  \end{pmatrix} \, ,
\end{equation*}
implying that when dealing with strongly correlated instruments, the efficient prices can be built by summing algebraically the volumes traded on each of them before normalizing by the global liquidity.
On the other hand, the bet of the IT \comment{up to first order in $\delta$} can be written as:
\begin{equation*}
  \bq = \frac12\sqrt{\frac{1}{\omega_{1} + \omega_{2}}}
  \begin{pmatrix}
  \omega_{1} (\Delta v_1 + \Delta v_2) \\
  \omega_{2} (\Delta v_1 + \Delta v_2)
\end{pmatrix} +
\sqrt{\frac{\omega_{1} \omega_{2}}{2(\omega_{1} + \omega_{2})}}
  \begin{pmatrix}
  \frac{1}{\sqrt{\delta}}(\Delta v_1 - \Delta v_2)\\
  \frac{1}{\sqrt{\delta}}(\Delta v_2 - \Delta v_1)
\end{pmatrix}\, .
  \end{equation*}
In this case the IT is encouraged to place orders proportional to either the average price variation, or orders of finite size proportional to the relative difference, rescaled by the typical size of the relative price moves $\sqrt{\delta}$. Indeed, the IT should bet on price variations inversely proportional to how common they are.
\end{example}

The specialization of the result above to the extreme case  $\rho \to \pm 1$ regime is particularly interesting, because it is a consequence of a more general result that can be applied whenever $\siginf$ is rank one, which is shown in the next example.

\begin{example}[Rank one price covariance]\label{ex:rank1} 
\comment{Let us consider the case in which all eigenvalues of $\siginf$ except for one tend to 0, and therefore $\siginf$ tends to a rank 1 matrix.} The interest of this example resides in the fact that $\lam$ displays in this case a particularly illuminating expression, providing a simple recipe for pooling together volumes belonging to different financial instruments correlated to the same underlying product.
Without loss of generality, $\siginf$ in the rank one case can be written as $\siginf = \bs \sigma \bs^\top$where $\bs \in \mathbb{R}^n$ is a unit vector and $\sigma > 0$. It is then straightforward to verify that the matrix
\begin{equation}
  \lam = \frac 1 2 \bs \left( \frac{\sigma}{\bs^\top \om \bs} \right)^{1/2} \bs^\top
\end{equation}
yields the  linear equilibrium of the model.
The above equation has a very simple interpretation:
\begin{itemize}
  \item The matrix $\lam$ commutes with $\siginf$, and its eigenvectors are insensitive to those of $\om$ (i.e., the volume fluctuations induced by noise traders).
  \item The factor $\sigma^{1/2}$ sets the scale for the price variations as being that  of the only mode of $\siginf$ that has non-zero fluctuations.
  \item By writing the principal component decomposition of $\om$ as $\om = \sum_{a=1}^n \bw_a \omega_{a} \bw^\top_a$, one can write the denominator as:
  \begin{equation}
    (\bs^\top \om \bs)^{1/2} = \left(\sum_{a=1}^n (\bw^\top_a \bs)^2 \omega_{a}\right)^{1/2} \, .
  \end{equation}
  Such a denominator sets the scaling with volume of the impact $\lam$. It amounts to a projection of  $\om$ on the only non-zero mode of $\siginf$, or equivalently a sum of the individual volume variances $\omega_{a}$ of $\om$, weighted by their projections on $\bs$.
\end{itemize}

\end{example}


\section{Implications}
\label{sec:implications}

The expression~\eqref{eq:lambda} for the impact matrix $\lam$ is of interest beyond the context of the Kyle Model. It  provides an insightful inference prescription for cross-impact models, as opposed to what has been done so far in the context of cross-impact fitting (see~\cite{benzaquen2017dissecting,mastromatteo2017trading,wang2016cross,schneider2016cross}). In such a context, one is interested in estimating from empirical data a model to forecast a price variation $\Delta\bp$ with a predictor of the form:
\begin{equation}
  \label{eq:ximp_estimation}
  \hat\Delta \bp = \hat \lam \bu \, ,
\end{equation}
in which it is essential to faithfully model how trading the instrument $i$ will impact the instrument $j$. The goal of this section is to establish a link between the Kyle model and the empirical calibration of Eq.~\eqref{eq:ximp_estimation}, where the price variations $\Delta\bp$ and the imbalances $\bu$ are sampled from empirical data.

\subsection{From theory to data: empirical averages and loss function}
First, let us consider a dataset of $T$ i.i.d.\ efficient price variations and volumes, $\{\Delta\bp^{(t)}\}_{t=1}^T$ and $\{\bu^{(t)}\}_{t=1}^T$, sampled from an unknown distribution with bounded variance that needs not to be related to the Kyle model presented in the previous sections of the manuscript. Accordingly, with a slight abuse of notation, from now on $\avg[ \cdot ]$ denotes the averages taken with respect to the underlying distribution from which $\Delta\bp$ and $\bu$ are sampled.
In addition, we introduce an empirical measure $\langle \cdot \rangle$ to denote the averages with respect to the empirical sample under investigation. In particular,  simple empirical estimators for the average price variation and the mean imbalance write:
\begin{align}
  \hat {\bp}_0 &=  \langle \Delta \bp \rangle = \frac 1 T \sum_{t=1}^T \Delta\bp^{(t)} \\
  \hat \bu_0 &=  \langle \bu \rangle = \frac 1 T \sum_{t=1}^T \bu^{(t)} \, .
\end{align}
Henceforth, we  consider that price changes and order flows are shifted by their empirical mean $\hat {\bp}_0$ and $\hat \bu_0$, and therefore we set $\hat {\bp}_0 =\hat \bu_0 =0$. The corresponding covariances are then defined accordingly as:
\begin{align}
  \hatsig &= \langle \Delta\bp \Delta\bp^\top\rangle \, , \\
  \hatomd &= \langle \bu \bu^\top\rangle \, , \\
  \hatrespinfd &= \langle \Delta\bp \bu^\top\rangle \, .
\end{align}
As $T\to\infty$, the empirical averages converge to the actual means and covariances.
Our goal is to show different recipes for the calibration of the model~\eqref{eq:ximp_estimation}, all relying on the estimators above. We shall compare such estimators with the \emph{Kyle estimator} $\hat \lam_{\rm Kyle}$, which, as we shall see,  displays several interesting properties. In order to evaluate the quality of different calibrations, we consider a quadratic loss $\chi^2$  defined below.
\begin{definition}[Loss]
Given an empirical measure $\langle \cdot \rangle$, we define the \emph{loss} $\chi^2$ as the function
\begin{equation}
  \label{eq:chisq}
  \chi^2 = \frac 1 2 \left< (\hat \Delta \bp - \Delta\bp)^\top \bm M(\hat \Delta \bp - \Delta\bp)\right> \, ,
\end{equation}
where $\hat \Delta \bp = \hat\lam \bu$ is a linear predictor of the efficient price variation and $\bm M$ is a SPD matrix.
\end{definition}
This definition of the loss implicitly implies that the calibration of the model is very similar to the task of the market maker in the Kyle model: whereas in the latter  the job of the MM is to forecast what the \emph{fundamental} price variation is on the basis of the order flow imbalance, in this setting one is required to forecast the \emph{efficient} (observed) price on the basis of the imbalance.
Even though the two problems are different, they involve extremely similar equations, a fact that will allow us to leverage the results of the previous sections, valid in principle for a Kyle Market Maker, in the calibration setting.
Hence, under this definition of loss one is only required to find proxies for the price variations $\Delta\bp$ and the dressed order imbalance $\bu$, disregarding in principle the realisations of the underlying fundamental price $\bpf$.

\subsection{Cross-impact estimators}
Here we present three different cross-impact estimators and discuss their properties. The proofs of the propositions in this section are provided in App.~\ref{app:estimation}.

\subsubsection{Maximum Likelihood Estimation}
\label{sec:linear_regression}
The simplest estimation recipe that one can construct in order to estimate the cross-impact matrix $\lam$ is probably the Maximum Likelihood Estimator (MLE):  $\hat\lam_{\rm MLE}$, that is obtained by minimising a quadratic loss.

\begin{proposition}\label{prop:chi2}
 The minimisation of the loss $\chi^2$ with respect to $\hat\lam_{\rm MLE}$  yields:
\begin{equation}
    \label{eq:mle}
\begin{aligned}
      \hat\lam_{\rm MLE} & =& \hatrespinfd (\hatomd)^{-1} \ ,
\end{aligned}
 \end{equation}
 independent of $\bm M$. The loss at the minimum is given by:
\begin{equation}
    \chi^2_{\rm MLE} = \frac12\tr\left(\bm M(\hatsig - \hat \sig_{\rm MLE})\right) \ ,
\end{equation}
where $\hat \sig_{\rm MLE} = \hatrespinfd(\hatomd)^{-1}(\hatrespinfd)^\top$ is the portion of the covariance of  fundamental price variations explained by the model~\eqref{eq:mle}.
\end{proposition}

Note that Eq.~\eqref{eq:mle} is almost identical to Eq.~\eqref{eq:MM_params}, that has been obtained by enforcing the efficient price condition. This is no coincidence, as in a linear equilibrium setting with normally distributed variables the traded price is exactly a linear function of the order flow imbalance $\bu$. There are though two important differences distinguishing the linear equilibrium of the Kyle model and the minimization of $\chi^2_{\rm MLE}$:
\begin{itemize}
\item A MM in the Kyle setting is actually performing a linear regression of the \emph{fundamental} price, whereas here we are interested in regressing the \emph{efficient} price. This justifies why Eq.~\eqref{eq:MM_params} uses the \emph{fundamental} price response $\respinf^d$, as opposed to Eq.~\eqref{eq:mle} that uses the efficient price $\resp^d$.
\item In the Kyle setup case, the MM is aware that a part of the order imbalance comes from the IT who owns information about the fundamental price and is optimising his strategy assuming that the MM uses a MLE.
  This leads to a quadratic equation for $\lam$ (see Eq.~(\ref{eq:self_consistent_lambda}) below) even in the linear equilibrium setting, whereas Eq.~\eqref{eq:mle} above is a linear relation for $\hat\lam_{\rm MLE}$.
\end{itemize}

\begin{remark}
One could have defined the Multivariate Kyle Model by replacing the efficient-price condition (see Eq.~\eqref{eq:MMcondition}) with a condition on the minimisation of $\chi^2$ without affecting the linear equilibrium. Moreover, this approach would do a better job at   generalising the model to more exotic settings. In fact, Eq.~\eqref{eq:mm_equation_general} provides a linear predictor only under the assumption of normality for prices and volumes: for other distributions, the relation linking $\bu$ and $\bpf$ is in general non-linear. Having a MM that minimises a loss $\chi^2$ is a reasonable assumption if one thinks that a MM without enough knowledge of the underlying distributions or computational power to build an efficient price should rely on linear regressions in order to estimate the  latter.
\end{remark}

The expression for $\hat \lam_{\rm MLE}$ only contains the dressed empirical estimators $\hatrespinfd$ and $\hatomd$, allowing one to estimate the impact matrix from real data.
In addition, such an estimator has the benefit of being very simple to implement, as it only requires the solution of a linear equation for $\hat \lam_{\rm MLE}$. Unfortunately, this estimator lacks several properties that turn out to be very useful in cases of practical interest:
\begin{description}
\item[Symmetry] The matrix $\hat \lam_{\rm MLE}$ is symmetric if and only if $\hatrespinfd$ and $\hatomd$ commute.
\item[Positive definiteness] The matrix $\hat \lam_{\rm MLE}$ can have negative eigenvalues (see \cite{benzaquen2017dissecting}).
\item[Consistency of correlations] In general, $[\hatsig, \hat \sig_{\rm MLE}] \neq 0$. Hence, the price variations inferred by using the order flow imbalance do not share the eigenvectors of the real price variations, unless the response $\hatrespinfd (\hatomd)^{-1} \hatrespinfd$  commutes with $\hatsig$.
\end{description}
The first two properties are extremely important in the calibration of cross-impact models, as shown in~\cite{alfonsi2016multivariate,schneider2016cross}: to perform pricing within a cross-impact setup, the matrix $\lam$ should be SPD in order to ensure absence of price manipulation. This is not the case for a MLE, which is thus not suitable for practical purposes.

\subsubsection{EigenLiquidity model}
In order to cure the lack of symmetry of the Maximum Likelihood Estimator, the idea proposed in Ref.~\cite{mastromatteo2017trading} is to construct an estimator of cross-impact $\hat \lam_{\rm ELM}$ that is symmetric by construction, by enforcing the relation:
\begin{equation}
  \label{eq:commutation}
  [\hat \lam_{\rm ELM}, \hatsig] = 0 \, ,
\end{equation}
which means that the impact eigendirections coincide with the eigenportfolios (or principal components of the asset space). This  prevents price manipulations that would be induced by an asymmetric of $\lam$.
Its calibration is explained in the next proposition.
\begin{proposition}\label{prop:chi2ELM}
Consider a pricing rule $\hat \Delta \bp = \hat\lam_{\rm ELM} \bu$, in which we impose the commutation relation~\eqref{eq:commutation}. Then, the maximum likelihood estimator obtained under such constraint takes the form:
  \begin{equation}
\begin{aligned}
    \hat \lam_{\rm ELM} &=& \sum_{a=1}^n \hat \bs_a g_a  \hat \bs_a^\top \, ,
\end{aligned}
  \end{equation}
  where $\{ \hat \bs_a \}_{a=1}^n$ are the empirical eigenvectors of $\hatsig$, and where:
  \begin{equation}
    g_a = \frac{\hat \bs_a^\top \hatrespinfd \hat \bs_a}{\hat \bs_a^\top \hatomd \hat \bs_a} \, .
  \end{equation}
Furthermore, the loss at the minimum is given by:
\begin{equation}
  \chi_{\rm ELM}^2 = \frac12\tr\left[M\left(\hatsig -\hatrespinfd\sum_{a=1}^n \hat\bs_a \frac{\hat\bs_a^\top\hatrespinfd\hat\bs_a}{\hat\bs_a^\top\hatomd\hat\bs_a}\hat\bs_a^\top\right)\right]\ .
\end{equation}
\end{proposition}
In this case, the properties discussed above become:
\begin{description}
\item[Symmetry] The matrix $\hat \lam_{\rm ELM}$ is symmetric by construction.
\item[Positive definiteness] The matrix $\hat \lam_{\rm ELM}$ can still have negative eigenvalues, although empirically the matrix $\hat\lam_{\rm ELM}$ has been reported not to display eigenvalues significantly smaller than zero (see~\cite{mastromatteo2017trading}).
\item[Consistency of correlations] Similarly to the MLE case, the price variation covariances $\hat\sig_{\rm ELM} = \hat\lam_{\rm ELM}\hatomd\hat\lam_{\rm ELM}$ do not generally commute with $\hatsig$. It is the case if and only if $[\hatomd, \hatsig] = 0$.
\end{description}
Summarizing, the price to pay in order to have a symmetric estimator is a larger loss function.
Note that if $\hatrespinfd$ and $\hatomd$ commute with $\hatsig$, then $\hat\lam_{\rm ELM} = (\hatrespinfd)^2(\hatomd)^{-1} = \hat\lam_{\rm MLE}$ and therefore the loss is equal to the best possible $\chi_{\rm ELM}^2 = \chi_{\rm MLE}^2$.


\begin{remark}[Estimators in one dimension]
Since commutation and symmetry are granted in one dimension, imposing the condition \eqref{eq:commutation} does not add any constraints to the minimization of the loss function and therefore $\hat\lambda_{\rm ELM} = \hat\lambda_{\rm MLE}= \frac{r^d}{\omega^d}$.
\end{remark}

\subsubsection{Kyle estimator}
Now, let us provide some intuition on how to construct an impact estimator inspired by Kyle's model, $\hat\lam_{\rm Kyle}$, by characterising the $\lam$ in the Kyle model from a different perspective.
First, take the MM solution Eq.~\eqref{eq:MM_params} (matching the MLE, Eq.~\eqref{eq:mle}), and exchange the dressed order imbalancees and responses with those predicted in the linear equilibrium through Eqs.~\eqref{eq:equil_resp} and~\eqref{eq:equil_omega}:
\begin{eqnarray}
  \respinfd &\to& \frac 1 2 \siginf \lam^{-1} \, , \\
  \omd &\to& 2 \om \ .
\end{eqnarray}
Then, the linear regression of Eq.~\eqref{eq:MM_params} becomes a quadratic matrix equation:
\begin{equation}
\label{eq:self_consistent_lambda}
\lam = \respinfd (\omd)^{-1} \to\lam =  \frac 1 4 \siginf \lam^{-1} \om^{-1} \, .
\end{equation}
After imposing symmetry and positive definiteness of $\lam$, this is exactly the equation that leads to the expression of $\lam$ obtained in the Kyle linear equilibrium (see App.~\ref{app:explicit_linear_equil}).
By doing so, one ensures an efficient price variation process whose associated covariance is half that of the ``true'' price variations. As alluded to above, in the Kyle setup the
MM is implicitly performing a linear regression, and \emph{knows that the IT is aware of his/her algorithm}, finally leading to Eq.~\eqref{eq:self_consistent_lambda}.
We now propose an estimator $\hat \lam_{\rm Kyle}$ that exploits these ideas.
\begin{proposition}\label{prop:chi2Kyle}
  Let us consider an estimator $\hat \Delta \bp := \hat \lam_{\rm Kyle}\bu$ 
 such that
  \begin{itemize}
  \item $\hat \lam_{\rm Kyle}$ is SPD
  \item The empirical covariance of the associated efficient price variations $\hat\sig_{\rm Kyle} = \langle \hat \Delta \bp \hat \Delta \bp^\top\rangle$ satisfies
    $$\hat\sig_{\rm Kyle}   = k^2 \hatsig$$
  \end{itemize}
  with $k \in \mathbb{R}$.
  Then the unique $\hat \lam_{\rm Kyle}$ satisfying these constraints is given by 
  \begin{equation}
    \label{eq:kyle_est}
  \hat \lam_{\rm Kyle} = k   (\hat\omright^{\rm d})^{-1} \sqrt{\hat \omright^{\rm d} \hatsig \hat \omleft^{\rm d}}  (\hat\omleft^{\rm d})^{-1}    
  \end{equation}
  where $\hatomd = \hat \omleft^{\rm d} \hat \omright^{\rm d}$ is a decomposition of $\hatomd$ such that $\hat \omleft^{\rm d} = (\hat \omright^{\rm d})^\top$.\\
 Moreover, the loss at the minimum is 
\begin{equation}
    \chi^2_{\rm Kyle} = \frac12\tr\left[\bm M \left((1 + k^2)\hatsig - 2k\hatrespinfd(\hat\omright^{\rm d})^{-1} \sqrt{\hat \omright^{\rm d} \hatsig \hat \omleft^{\rm d}}  (\hat\omleft^{\rm d})^{-1}    \right)\right]
\end{equation} 
and the value of $k$ that minimizes the loss is given by 
\[k^\star = \frac{ \tr(\bm M\hatrespinfd(\hat\omright^{\rm d})^{-1} \sqrt{\hat \omright^{\rm d} \hatsig \hat \omleft^{\rm d}}  (\hat\omleft^{\rm d})^{-1} )}{\tr(\bm M\hatsig)}\ .\]
\end{proposition}

Note that the complete analogy with the Kyle model is recovered for $k = 1$. In fact, if the sample dataset is drawn from an actual Multivariate Kyle Model with $\siginf / 2 = \sig=\hatsig$ and $2 \om = \omd =\hatomd$ as described in the previous sections, one obtains: 
$$
\avg[\hat \lam_{\rm Kyle}] \xrightarrow{T\to\infty}  k \lam \, .
$$
The advantage of leaving $k$ as a free parameter is that it allows  to increase or decrease the loss $\chi^2$ of the Kyle estimator without affecting the eigenvectors of the predicted covariance. In fact, the relationship $2\sig = \siginf $ arising in the Kyle model is thought not to be universal, so that it is more natural to think of it as a phenomenological parameter expressing the informational content of trades.\\

The adoption of $\hat\lam_{\rm Kyle}$ as an estimator of cross-impact has several interesting advantages with respect to more customary estimators such as the Maximum Likelihood Estimator $\hat\lam_{\rm MLE}$ or an impact estimator based on the Eigen-Liquidity Model $\hat\lam_{\rm ELM}$:
\paragraph{Symmetry and positive definiteness}

The Kyle construction imposes \emph{a priori} symmetry and PDness due to the request of optimality of the IT. In fact, the optimality condition~\eqref{eq:ITstrategy} and the stability one (PDness of $\lam$) are exactly the ones that a risk-neutral rational agent, as considered in~\cite{alfonsi2016multivariate,schneider2016cross}, would face when trying to optimise his profits.

\paragraph{Consistency of correlations}
This construction allows to recover the empirical correlations $\hatsig$ for $k=1$. For \emph{any value} of $k$ one has $[ \hatsig, \hat\sig_{\rm Kyle} ] = 0$.

\paragraph{Loss}
An important point concerns the loss function obtained by using $\hat\lam_{\rm Kyle}$. In general $\chi^2_{\rm Kyle}$ is \emph{larger} than that obtained by doing Maximum Likelihood Estimation, which minimizes $\chi^2$ by design. 
Hence, the price to pay in order to have symmetry, positive semi-definiteness and consistency of correlations is in general a worse fit of empirical data with respect to $\hat\lam_{\rm MLE}$ (as measured by the loss $\chi^2$). However, if our data behaves as the Kyle model, i.e. if $\hatrespinfd = \hatsig \hat\lam_{\rm Kyle}^{-1}$ then $k^\star = 1$ and:
\[\chi^2_{\rm Kyle} = 0 = \chi^2_{\rm MLE} \ ,\]
implying that $\lam$ is (\emph{obviously!}) the best estimator for a market that reproduces the statistics of the multivariate Kyle model.
Moreover, the loss at the minimum is zero because the \emph{efficient} price variation $\Delta \bp$ is completely determined by the imbalance. Indeed, the problem of regressing the \emph{fundamental} price variation $\Delta \bpf$ would have yielded a non-zero loss at the minimum even for data sampled from an actual Kyle model due to the relation $\sig < \siginf$.

\begin{remark}[ELM and Kyle estimators for rank one $\hatsig$]
 As we showed in example \ref{ex:rank1}, when the price variation covariance is rank one $\hat\lam_{\rm Kyle}$ is proportional to $\hatsig$. On the other hand, $\hat\lam_{\rm ELM}$ has, by definition, the same eigenbasis as the price variations covariance. Therefore, if $\hatsig$ is rank one then $\hat\lam_{\rm ELM}$ is proportional to $\hat\lam_{\rm Kyle}$ and $\hatsig$.
\end{remark}

\subsection{A comparison of recipes}\label{sec:empirical}

In order to compare the performance of the estimators we compute observables that are directly related to the properties discused above:
\begin{itemize}
 \item \textbf{Loss:} we compute the loss $\chi^2$ as given in Eq.~\eqref{eq:chisq}.
 \item \textbf{Asymmetry:} to quantify the degree of symmetry of the estimator we compute the norm of its antisymmetric part divided by the total norm:
\[\alpha = \frac{|\hat\lam - \hat\lam^\top|}{2|\hat\lam|}\ .\]
  The value $\alpha = 0$ ($\alpha = 1$) means that the estimator is symmetric (antisymmetric).
  \item \textbf{Positive definiteness:} we compute the lowest real part of the spectrum of the estimator:
\[\lambda^\star = \min_i\Re(\hat\lambda_i)\ ,\]
where $\{\hat\lambda_i\}_{1\geq i \geq n}$ are the eigenvalues of $\hat\lam$. Positive definiteness is equivalent to $\lambda^\star > 0$.
\item\textbf{Consistency of correlations:} we compute the norm of the commutator of $\hatsig$ and the covariances of the estimated price variations $\hatsig_{\rm est}$:
\[\kappa = \frac{|\hatsig_{\rm est}\hatsig - \hatsig\hatsig_{\rm est}|}{|\hatsig||\hatsig_{\rm est}|}\ .\]
Correlations are consistent if $\kappa = 0$.
\end{itemize}

\subsubsection{Application to real data}

In Fig.~\ref{fig:all} we show the covariances and estimators for the 2-year, 5-year, 10-year, and 30-year tenors of the U.S. Treasury Futures traded in the Chicago Board of Trade (see~\ref{sec:fabricating} for details). In Table \ref{tb:observables} we show the values of the observables corresponding to the $4\times4$ system. In Table \ref{tb:6pairs} we show the values of the same obserbavles but averaged over the 6 possible combinations of $2\times2$ systems. As expected, the Kyle procedure fares slightly worse for the loss function, but is much better than other procedures on all other counts.

\begin{table}[h]
\centering
\begin{tabular}{ c c c c c }
\hline
Estimator & Loss ($\chi^2$) & Commutator ($\kappa$) & PDness ($\lambda^\star$) & Asymmetry ($\alpha$)\\
\hline
$\hat\lam_{\rm MLE}$  & 1.20 & 0.038 & 0.031 & 0.31\\
$\hat\lam_{\rm ELM}$& 1.26 & 0.129 & 0.017 & 0\\
$\hat\lam_{\rm Kyle}$& 1.32 & 0 & 0.386 &0\\ \hline
\end{tabular}
\caption{Values of the observables for the three estimators for U.S. Treasury Futures and Bonds.}\label{tb:observables}
\end{table}

\begin{figure}[H]
\centering
 \includegraphics[width=\textwidth]{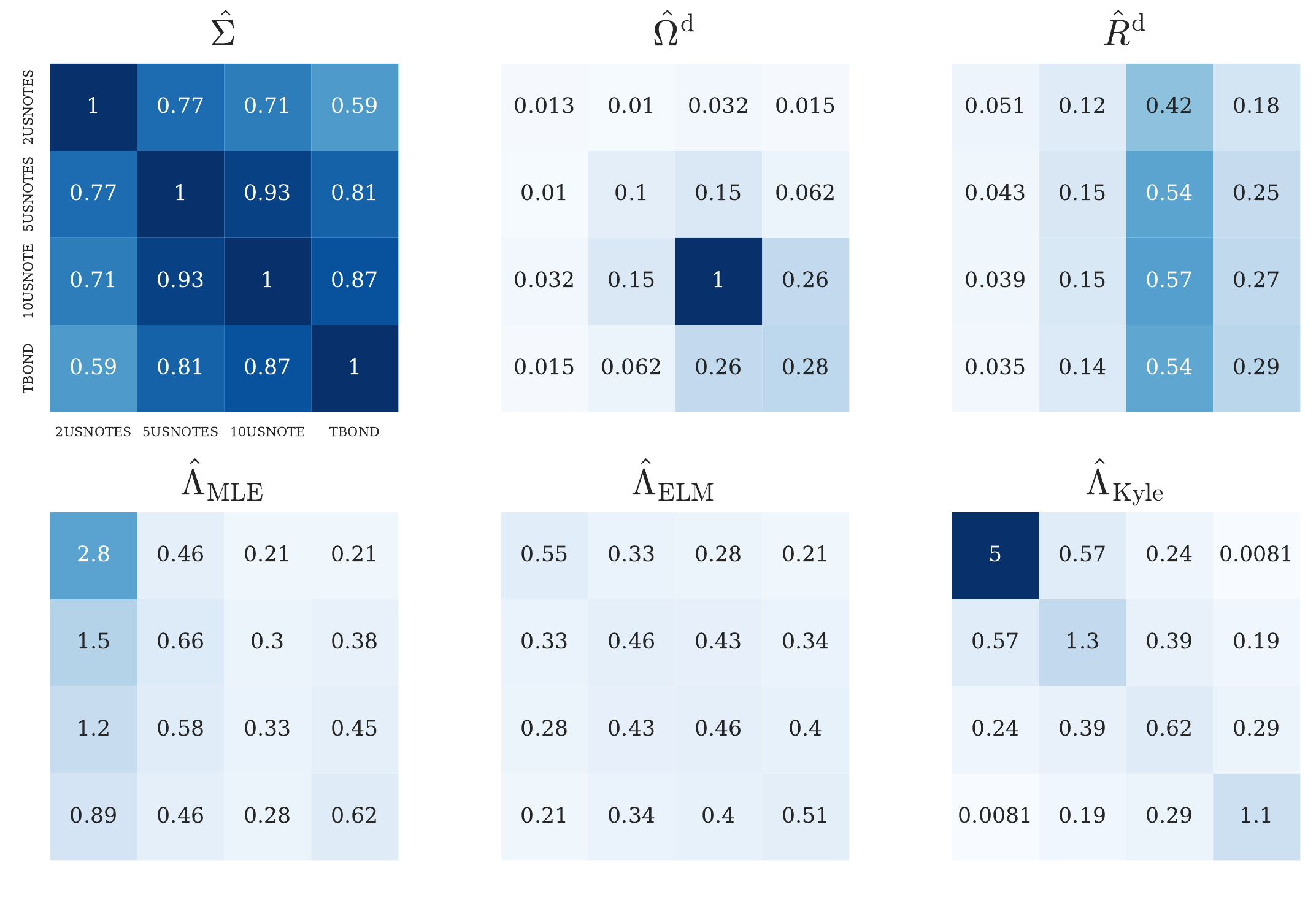}
 \caption{\textbf{Averaged US Treasury Futures covariances and impact estimators for 2016.}\newline
 Top from left to right: Daily averaged price variation covariances in untis of risk squared. Unitless daily averaged volume covariances rescaled with respect to the maximum. To obtain the original volumes $\hatomd$ has to be multiplied by $2242\$^2$. Market response rescaled with respect to the maximal volume in units of risk. To obtain the original volumes $\hatrespinfd$ has to be multiplied by $\sqrt{2242}\$$. \newline
Bottom from left to right: Maximum Likelyhood impact estimator, EigenLiquidity Model based impact estimator and Kyle model based impact estimator in units of risk. To obtain the estimators in the right units they have to be divided by $\sqrt{2242}\$$.}\label{fig:all}
\end{figure}

\begin{table}[h] 
\centering 
\begin{tabular}{ c c c c c } 
\hline
 Estimator & Loss ($\chi^2$) & Commutator ($\kappa$) & PDness ($\lambda^\star$) & Asymmetry ($\alpha$)\\ 
\hline 
$\hat\lam_{\rm MLE}$  & 0.668 & 0.033 & 0.318 & 0.186\\ 
$\hat\lam_{\rm ELM}$& 0.691 & 0.108 & 0.231 & 0\\ 
$\hat\lam_{\rm Kyle}$& 0.718 & 0 & 0.952 &0 \\
\hline
\end{tabular} 
\caption{Average values of the observables for the three estimators for the 6 pairs of U.S. Treasury Futures.\label{tb:6pairs}} 
\end{table}

\subsubsection{Synthetic data} 
To further explore the effects of price variation and liquidity correlations, in the next two sections we compute the values of the observables described above on the 6 sets of synthetic $2\times2$ covariance matrices with specific price variation covariances $(\hatsig_\rho,\hatomd_\rho,\hatrespd_\rho)$ and specific volume covariances $(\hatsig_\epsilon,\hatomd_\epsilon,{\hatrespd}_\epsilon)$ respectively fabricated by modifying the data for U.S. Treasury Futures as explained in~\ref{sec:fabricating} (see Table~\ref{tb:6pairs}).
The idea behind this construction is to provide a synthetic but realistic set of matrices that are parametrized by either a liquidity parameter $\epsilon$ or a correlation parameter $\rho$. By varying
those parameters, one can extrapolate from the reality (recovered for specific values of $\epsilon$ and/or $\rho$) and a fictitious world in which assets can be made more or less liquid with respected to reality by varying the knob $\epsilon$, and more or less correlated by varying the parameter $\rho$.

\begin{figure}[t!]
\centering
 \includegraphics[width=\textwidth]{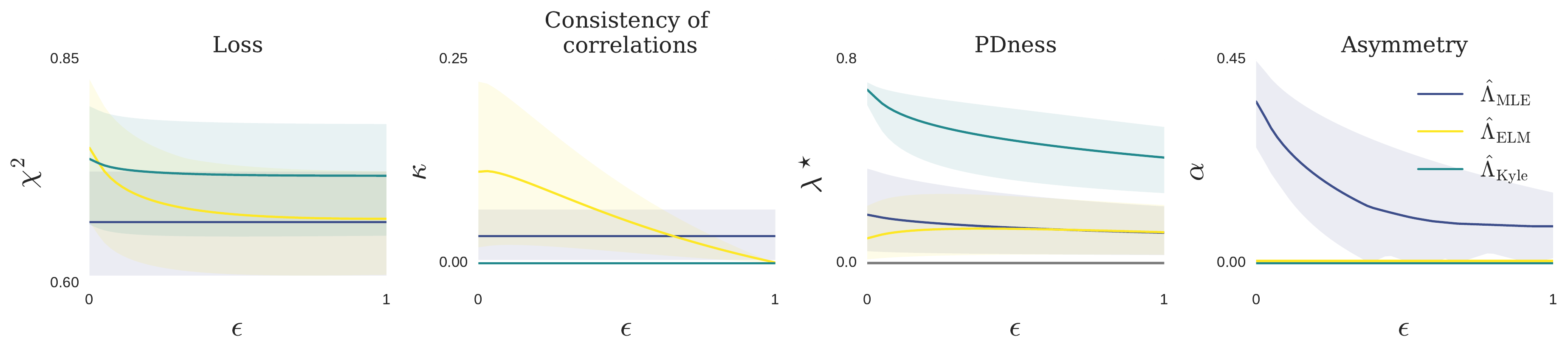}
 \caption{\textbf{Effect of liquidity.} Averaged values (solid lines) and intervals between the maximum and the minimum (shaded area) of the observables for $\epsilon$ ranging from 0 to 1.}\label{fig:liquidity}
\end{figure}

\paragraph{Extreme illiquidity}
In order to explore the effect of extreme heterogeneous liquidities we rescale the $2\times2$ covariances as:
\[\hatomd\mapsto\hatomd_\epsilon = \begin{pmatrix}
             1 & \sqrt{\epsilon}\,\hatomd_{12}\\
	    \sqrt{\epsilon}\,\hatomd_{12} & \epsilon
            \end{pmatrix}\ .
\]
In Fig.~\ref{fig:liquidity} we show the value of each of the observables averaged over the 6 sets of covariances for $\epsilon$ ranging from $0$ to $1$. Regarding the loss, the estimator that better performs is $\hat\lam_{\rm MLE}$, regardless of $\epsilon$ as expected. It is interesting to observe that $\chi^2_{\rm MLE}$ does not depend on $\epsilon$. The reason is that, by construction, the estimated price variation covariances are invariant under changes in the volume variances\footnote{This is not a particular property of systems in 2 dimensions, it holds for all dimensions.}:
\[\hatsig_{\rm MLE} = \hatrespd_\epsilon(\hatomd_\epsilon)^{-1}(\hatrespd_\epsilon)^\top = \hatrespd(\hatomd)^{-1}(\hatrespd)^\top.\]
The loss of the other two estimators depends on $\epsilon$. For small values of $\epsilon$, $\chi^2_{\rm Kyle}$ is lower than $\chi^2_{\rm ELM}$, suggesting that $\hat\lam_{\rm Kyle}$ might be better in case of markets with heterogeneous volatilities. For $\epsilon\gtrsim0.2$ the loss distributions for the three estimators have a lot of overlap and no one is significantly better.
\begin{figure}[t!]
\centering
 \includegraphics[width=0.7\textwidth]{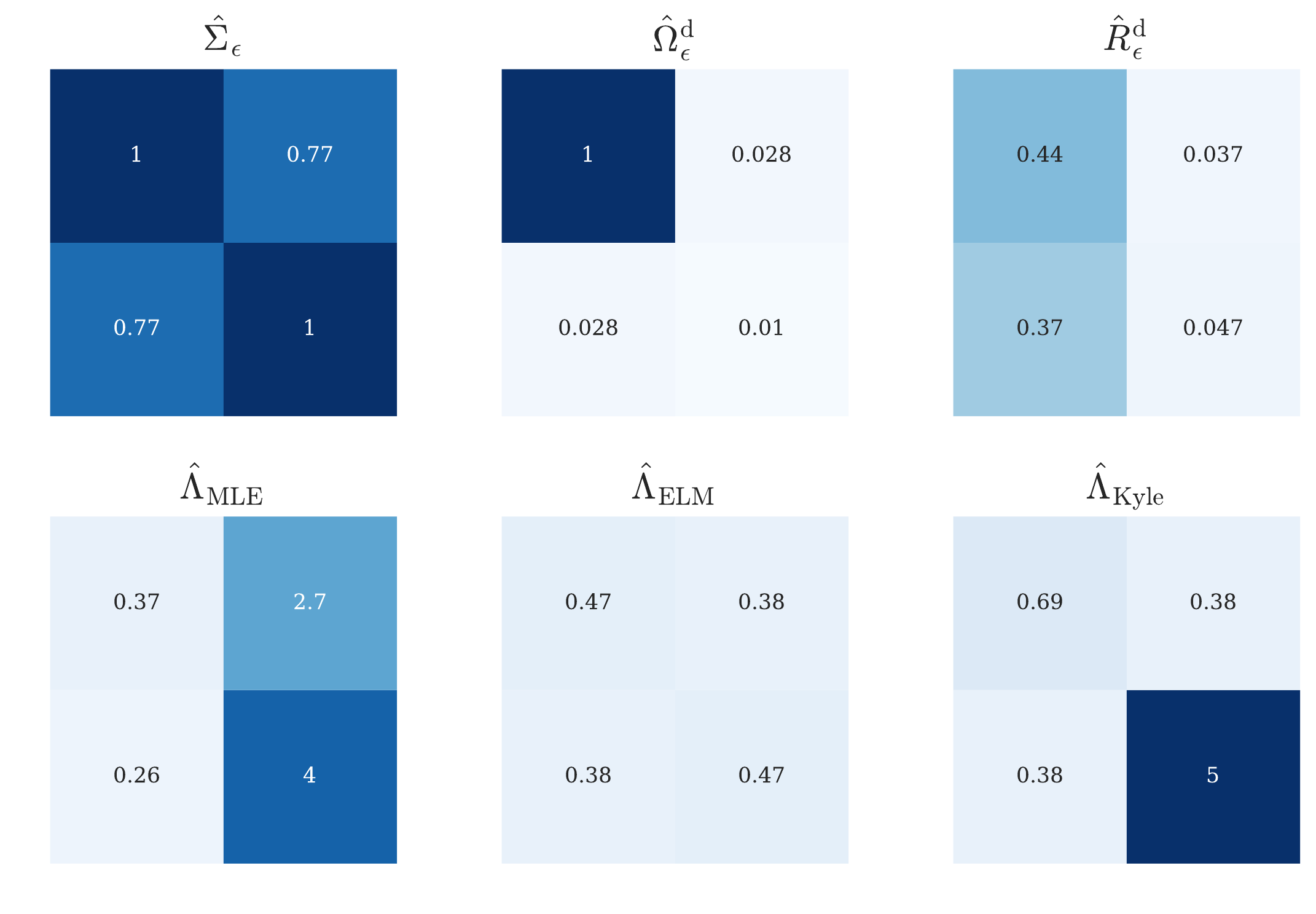}
 \caption{\textbf{Extreme illiquidity.} Covariance matrices and estimators with an illiquid asset. The value of $\rho=0.77$ is the empirical correlation of the representative pair 2USNOTES-5USNOTES, whereas the value of $\epsilon=0.01$ is obtained by rescaling the empirical responses and volume covariances of the same pair. }\label{fig:illiquid}
\end{figure}
As explained above, $\kappa_{\rm MLE}$ does not depend on $\epsilon$ and $\kappa_{\rm Kyle} = 0$, by construction. On the other hand, the limit $\kappa_{\rm ELM}$ decreases to 0 as $\epsilon\to1$ because in the limit $\epsilon = 1$, $\hatomd_\epsilon$ and $\hatsig_\epsilon$ commute.
All estimators are PD, and the only non-symmetric estimator is $\hat\lam_{\rm MLE}$, whose antisymmetric part decreases as $\epsilon$ increases (even though it never becomes fully symmetric). The other two are symmetric by construction.\\

Reminiscent to what we showed in example \ref{ex:extreme_illiquidity}, in Fig. \ref{fig:illiquid} we display the covariance matrices and the three corresponding impact estimators in a case where one of the assets' liquidity is much smaller than the other one. Note that due to the lack of liquidity of the second asset the second column of $\hatrespd_\epsilon$ is considerably smaller than the first one while price covariances are not changed. Let us analyze heuristically the  estimators obtained:
\begin{itemize}
 \item $\hat\lam_{\rm MLE}$: Trading the second asset will have a strong impact on its price and a significant impact on the price of first asset (due to the non-negligible price correlations). Trading the first asset will have little impact on both prices. 
 \item $\hat\lam_{\rm ELM}$: In this case the impact estimator predicts that trading either asset will have a similar impact.
\item $\hat\lam_{\rm Kyle}$: The main difference between the predictions of the Kyle estimator and  $\hat\lam_{\rm MLE}$ is that in this case trading the second asset will strongly modify its price but it will have very little impact on the first assets' price. As already emphasized, this is how efficient markets should work: liquid instruments should not be affected by the order flow on correlated illiquid instruments, otherwise price manipulation strategies would be possible. 
\end{itemize}

\begin{figure}[t!]
\centering
 \includegraphics[width=\textwidth]{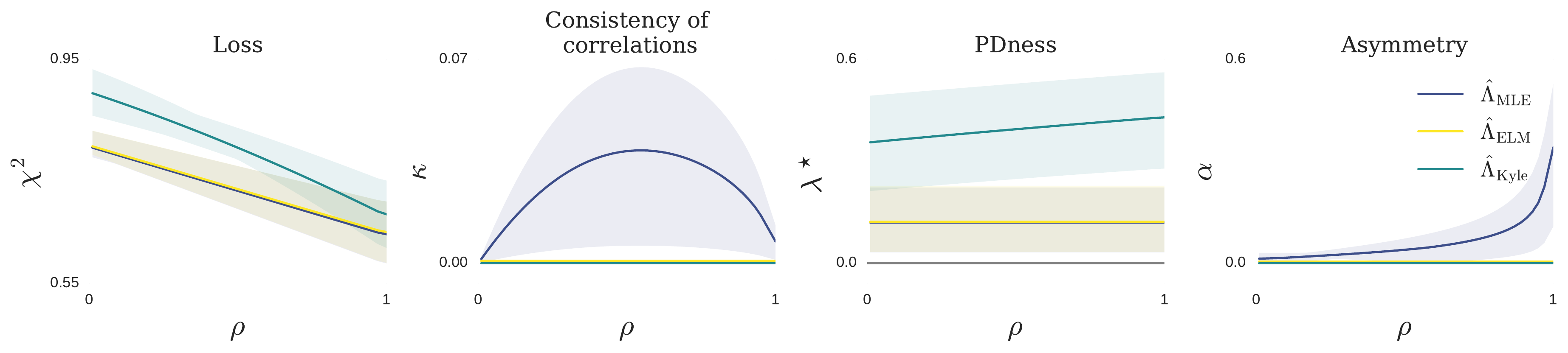}
 \caption{ \textbf{Effect of price variation correlations} Averaged values (solid lines) and intervals between the maximum and the minimum (shaded area) of the observables for $\rho$ ranging from 0 to 1.}\label{fig:correlations}
\end{figure}

\paragraph{Extreme correlations}
We now look at a linear combinations of two of the assets with price variation covariances: 
\[\hatsig\mapsto\hatsig_\rho = \begin{pmatrix}
             1 & \rho\\
	    \rho & 1
            \end{pmatrix}\ .
\]

In Fig.~\ref{fig:correlations} we show the value of each of the observables averaged over the 6 sets of covariances for $\epsilon$ ranging from $0$ to $1$. The loss of all estimators decreases as $\rho\to1$ and as in the previous case, the distributions have a strong overlap, suggesting again that no one estimator is significantly better at minimizing the loss. \\

By construction $\kappa_{\rm Kyle=0}$, and in this case $\kappa_{\rm ELM} = 0$ because $\hatomd_\rho$ commutes with $\hatsig_\rho$ (it would not be the case if the volume variances were not the same). Also $\kappa_{\rm MLE} = 0$ for $\rho = 0$ because in that case $\hatsig_\rho = \id$, but for $\rho > 0$, $\kappa_{\rm MLE} > 0$.
The distributions of $\lambda^\star_{\rm ELM}$ and $\lambda^\star_{\rm MLE}$ are essentially the same. $\lambda^\star_{\rm Kyle}$ is significantly larger than the other two for all values of $\rho$.
Asymmetries in $\hat\lam_{\rm MLE}$ are amplified as $\rho\to1$ because in that limit $\hatomd_\rho$ becomes singular.\\

Again, following example \ref{ex:extreme_correlation} in Fig. \ref{fig:rankone} we show the covariance matrices and the three corresponding impact estimators for a case where the price variations of the two assets are strongly correlated. Let us analyze the estimators obtained:
\begin{itemize}
 \item $\hat\lam_{\rm MLE}$: Small arbitrary fluctuations in $\hatrespd_\rho$ are amplified due to the strong volume correlations that result from the transformation. In the example shown in the figure the impact on the second asset is much larger than on the first one.
 \item $\hat\lam_{\rm ELM}$: Self impact and cross-impact are comparable. Furthermore, self-impact is very similar for both assets.
\item $\hat\lam_{\rm Kyle}$: Self-impact is very similar for both assets and stronger than cross-impact.
\end{itemize}

\section{Conclusion}

In the present work we studied a Multivariate Kyle model, that proved to be a very interesting setting to understand the fundamental mechanisms underlying {\it cross-impact}. Perhaps more importantly, the Multivariate Kyle model suggests a practical recipe to extract a consistent cross-impact matrix structure from empirical data -- a point that does not seem to have been emphasized before but that becomes crucial when dealing with present day large dimensional data. We recovered the Caball\'e-Krishnan solution at equilibrium and proved the unicity of the symmetric solution. We provided an interpretation of the results with regard to the eigen-modes of returns and volumes covariances. We discussed the implications of the model for pricing, with a particular focus on the SPD property of the impact matrix. We presented the implications for cross-impact regression of empirical data, and showed that cross-awareness (or in lesser terms ``I know that you know") can be used as powerful regulariser with small loss in predictive power. We confronted our results with previous empirical cross-impact analyses, and identified limiting regimes in which the results in~\cite{benzaquen2017dissecting,mastromatteo2017trading} are reproduced.

From a complementary point of view, our results can be seen as proxies for the behavior of an idealized market in which prices fully reflect the information encoded in the order flow: for a manipulation-free market, prices of liquid securities should be insensitive to trading of strongly illiquid instruments
(see example~\ref{ex:extreme_illiquidity}), whereas prices of strongly correlated instruments should be insensitive to how they are individually traded (example~\ref{ex:extreme_correlation}). Measuring how much real markets violate these principles would be an interesting empirical application of our results, that could be used to assist regulators in order to assess the vulnerability of a market to correlated trading activity.

In spite of all the ``good'' properties of the Kyle cross-impact estimator, one should keep in mind that many important aspects are left out of the Kyle framework, that may play a crucial role in practice. First of all, the empirical order flow is strongly autocorrelated in time, which must induce a non-trivial lag dependence of the impact function -- as found in \cite{benzaquen2017dissecting}. Extending the present theory to predict a lag-dependent impact matrix $\lam(\tau)$ would be extremely useful. Second, impact is non-linear in traded quantities: it is now well accepted that the impact of a metaorder has a square-root dependence on volume. How this single asset square-root law generalizes in a multiasset context is, as far as we are aware, a completely open problem. Finally, it would be interesting to extend to the multivariate case some recent extensions of the Kyle model that accounts for the inventory risk of the market maker \cite{Cetin}. We hope to visit some of these questions in future work.

\begin{figure}[t!]
\centering
 \includegraphics[width=0.7\textwidth]{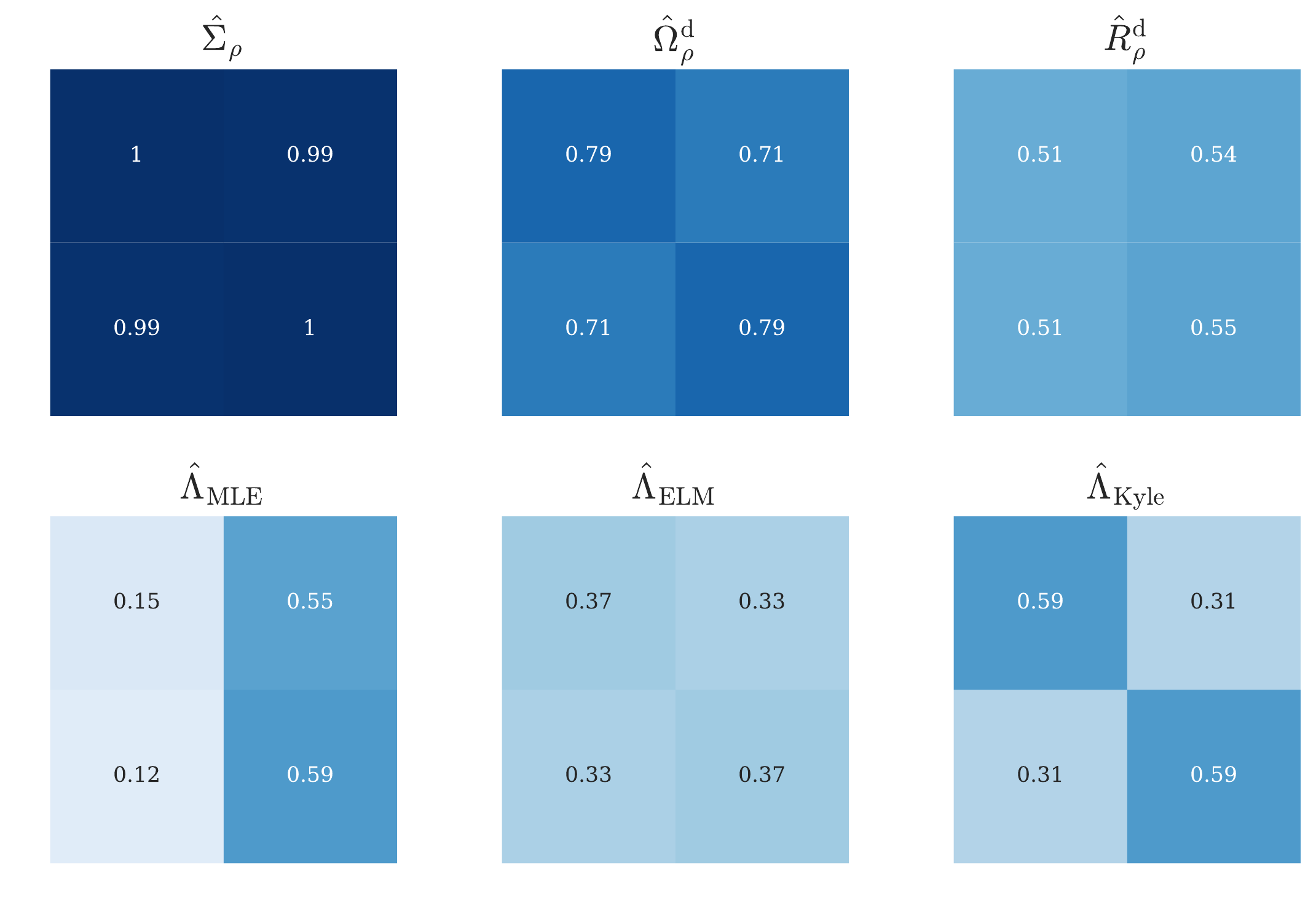}
 \caption{\textbf{Strong price variation correlations.} Covariance matrices and estimators with strongly correlated price variations ($\rho=0.99$)}\label{fig:rankone}
\end{figure}

\section*{Acknowledgments }

We thank Z. Eisler, A. Fosset and E. S\'eri\'e  for very fruitful discussions. We also thank J. Caball\'e for kindly providing the unpublished manuscript \cite{caballe1989insider}.


\clearpage

\bibliographystyle{alpha}
\bibliography{kyle}

\clearpage

\appendix
\renewcommand*{\thesection}{\Alph{section}}
\section{Solution of the Multivariate Kyle model}
\label{app:solution}

The proof of Theorem \ref{th:main_theorem} is split in 4 different results. First, the fact that the equilibrium is linear is a consequence of Propositions \ref{th:IT_strategy} and \ref{prop:MM_strategy}. Before finding the explicit form of the equilibrium we need to show that $\lam$ is symmetric. This is a consequence of the second order condition on the minimization of the IT's strategy (Proposition \ref{th:IT_strategy}) and is proven in Proposition \ref{prop:symmetry}. Finally, the explicit form and proof of uniqueness of $\lam$ is given in Proposition \ref{prop:unique_sym}

\subsection{Optimality of the Informed Trader}
\label{app:it_solution}
\noindent\emph{Proof of Proposition~\ref{th:IT_strategy}. }
Here, we show the consequences of the first and second order conditions arising from the requirement that the IT is maximizing locally the average utility given $\bpf$, which can be written as:
\begin{align*}
\avg[\util_{IT} | \bpf] &= -\avg[\bq^\top (\bp - \bpf) | \bpf]\\
&= -\avg[\bq^\top (\bmu + \lam\bu - \bpf) | \bpf]\\
&= -\bq^\top \bmu - \bq^\top\lam\bq + \bq^\top\bpf.
\end{align*}
In order to find the optimal strategy, we maximize the utility with respect to $\bq$. The first and seccond order derivatives are:
\begin{align}
 \label{eq:firstorder}\frac{\partial \avg[\util_{IT}|\bpf]}{\partial \bq} &= - \bmu - 2\lam_{\rm S}\bq  + \bpf = 0\\
 \label{eq:secondorder}\frac{\partial^2 \avg[\util_{IT}|\bpf]}{\partial \bq \partial \bq^\top} &= - 2\lam_{\rm S}.
\end{align}
where $\bm X_{\rm S}$ ($\bm X_{\rm A}$) denotes the symmetric (antisymmetric) part of $\bm X$.\\
Since $\util_{\rm IT}$ is quadratic on $\bq$, the profit maximization condition implies that the second derivative \eqref{eq:secondorder} has to be strictly negative definite, which in turn implies that both $\lam$ and $\lam_{\rm S}$ have to be PD.
Solving \eqref{eq:firstorder} for $\bq$ leads to
\begin{equation}
 \label{eq:q_solution_app}
 \bq = \ba + \bB\bpf
\end{equation}
where
\begin{equation}\label{eq:a}
\begin{aligned}
\ba &= -\bB\bmu \\
\bB &= \frac12\bm\Lambda_{\rm S}^{-1}.
\end{aligned}
\end{equation}
Since $\lam_{\rm S}$ has to be PD, it is invertible and $\bB$ and $\ba$ are well defined.
\qed\\

As we will show in App.~\ref{app:sym_lin_eq}, PDness of $\bm \Lambda$ together with the efficient-price condition (discussed in App.~\ref{app:mm_solution}) implies that $\bm\Lambda$ is symmetric. 

\subsection{Optimality of the Market Maker}
\label{app:mm_solution}
In this Appendix we will state the conditions stemming from the requirement that the MM fixes a pricing rule $\bp = \lam \bu + \bmu$ such that $\bp = \avg[\bpf | \bu]$ is an efficient price.\\

\noindent\emph{Proof of Proposition~\ref{prop:MM_strategy}. }
Exploiting the Gaussian nature of $\bpf$ and $\bv$ and  their independence, the efficient price as given in \eqref{eq:MMcondition} can be expanded as \cite[p.~269]{lindgren2013stationary}:
\begin{align}
 \nonumber\bp = \avg[\bpf|\bu] &= \avg[\bpf] + \cov[\bpf,\bu]\cov[\bu,\bu]^{-1}(\bu - \avg[\bu])\\
 \nonumber& = \bp_0 + \respinfd (\omd)^{-1}(\bu - \bu_0)\\
\label{eq:MMstrategy}& = \bmu + \bm \Lambda \bu \, ,
\end{align}
where
\begin{align}
\label{eq:mu_app} \bmu &= \bp_0 - \lam \bu_0 \, , \\
\label{eq:lambda_app}\bm\Lambda & = \respinfd (\omd)^{-1}.
\end{align}
\hfill\qed\\

This analysis provides an explicit form for the MM's strategy and it might look as if the MM is doing a linear regression. However, without knowing the strategy of the IT one cannot estimate $\respinfd$ nor $(\omd)^{-1}$ meaning that the MM has to include information about $\bq$ in his regression making it a quadratic problem rather than a linear one. This point is discussed in more detail in Sec.~\ref{sec:linear_regression}.

\subsection{Symmetry of the linear equilibrium}
\label{app:sym_lin_eq}

In order to find the linear equilibrium described in Theorem~\ref{th:main_theorem}, we first have to prove the proposition \ref{prop:symmetry}. To do so we will make use of an useful lemma \cite[Theorem 7.2.6]{horn1985matrix} that will be used in the following derivations.

\begin{lemma}[Positive definite square-root]
  \label{lemma:pos_def_sqrt}
  Consider a matrix $\bm Y$ that is symmetric and positive semi-definite. Then, it exists a unique symmetric, positive semi-definite matrix $\bm X = \sqrt{\bm Y}$ such that $\bm X \bm X = \bm Y$.
\end{lemma}


\noindent\emph{Proof of Proposition~\ref{prop:symmetry}.} Let us start by plugging $\avg[\bu] = \bu_0=\frac12\bm\lam_{\rm S}^{-1}(\bp_0 - \bmu)$ into \eqref{eq:MM_params}. After a bit of algebra we obtain
\begin{equation}
 (\bp_0 - \bmu)(\id - \frac12\bm\Lambda\bm\Lambda_{\rm S}^{-1}) = 0
\end{equation}
thus $\bmu = \bp_0$ and $\bq = \frac12\bm\Lambda_{\rm S}^{-1}(\bpf - \bp_0)$ which allows us to compute the following quantities:
\begin{align}
 \bu_0 &=0\\
\omd & = \frac14\bm\Lambda_{\rm S}^{-1}\siginf\bm\Lambda_{\rm S}^{-1} + \om\\
\respinfd & = \frac12\siginf\bm\Lambda_{\rm S}^{-1}.
\end{align}
By plugging the results above into \eqref{eq:MM_params} one obtains a quadratic equation:
\begin{equation}
(\bm\Lambda\bm\Lambda_{\rm S}^{-1} - 2\id)\frac14\siginf + \bm\Lambda\om \bm\Lambda_{\rm S} = 0.\label{eq:lambdafirst}
\end{equation}
Expanding $\bm\Lambda$ as $(\bm\Lambda_{\rm A} + \bm\Lambda_{\rm S})$ and doing some algebra from \eqref{eq:lambdafirst} we obtain
\begin{equation}\label{eq:lambda_a}
 \bm\Lambda_{\rm A} = (\id - \bm\Gamma)(\id + \bm\Gamma)^{-1}\bm\Lambda_{\rm S}
\end{equation}
where $\bm\Gamma = 4\bm\Lambda_{\rm S}\om \bm\Lambda_{\rm S}\siginf^{-1}$. Now, since \comment{$\bm\Lambda_{\rm A} = \frac12(\bm \Lambda - \bm\Lambda^\top)$, we have $\bm\Lambda_{\rm A} = -\bm\Lambda_{\rm A}^\top$ and therefore}
\begin{align*}
 (\id - \bm\Gamma)(\id + \bm\Gamma)^{-1}\bm\Lambda_{\rm S} &= -\bm\Lambda_{\rm S}(\id + \bm\Gamma^\top)^{-1}(\id - \bm\Gamma^\top)\\
(\id + \bm\Gamma^\top)\bm\Lambda_{\rm S}^{-1}(\id - \bm\Gamma) &= -(\id - \bm\Gamma^\top)\bm\Lambda_{\rm S}^{-1}(\id + \bm\Gamma)\\
\sqrt{\bm\Lambda_{\rm S}}\bm\Gamma^\top\bm\Lambda_{\rm S}^{-1}\bm\Gamma\sqrt{\bm\Lambda_{\rm S}}&=\id.
\end{align*}
Substituting back $\bm\Gamma$ we get
\begin{equation*}
\bm Y^2 = \bm Z^2
\end{equation*}
where $\bm Y = \sqrt{ \bm\Lambda_{\rm S}} \om\sqrt{\bm\Lambda_{\rm S}}$ and $\bm Z = \sqrt{\bm\Lambda_{\rm S}^{-1}} \frac14\siginf\sqrt{\bm\Lambda_{\rm S}^{-1}}$. Since  $\bm\Lambda_{\rm S}$ is PD, $\bm Y^2$ is also SPD and therefore $\bm Y=\bm Z$ is the unique PD square root of $\bm Y^2$ which leads to
\begin{equation}
 \frac14\siginf = \bm \Lambda_{\rm S}  \om \bm \Lambda_{\rm S}\label{eq:lambdaS}
\end{equation}
and substituting these results into \eqref{eq:lambda_a} reveals that $\lam_{\rm A} = 0$. \qed\\

With equation \eqref{eq:lambdaS} we prove that the profit maximization condition on the IT's cost (i.e. $\bm\Lambda$ has to be PD) results in $\bm\Lambda_{\rm A} = 0$ and $\bm\Lambda = \bm\Lambda_{\rm S}$. If $\bm\Lambda_{\rm S}$ was not PD, $\bm Y$ and $\bm Z$ could be different roots and we would not be able to establish the relationship $\bm Y = \bm Z$. In that case we would find many solutions, most of them being saddle points for $\avg[\util_{IT}]$.

\subsection{Explicit form of the linear equilibrium}
\label{app:explicit_linear_equil}

The last step needed in order to prove \comment{the existence and uniqueness} the linear equilibrium in Theorem~\ref{th:main_theorem} is to prove that Eq.~\eqref{eq:lambdaS} admits a unique symmetric solution, as made explicit below.
\begin{proposition}
\label{prop:unique_sym}
The unique symmetric PD matrix $\lam$ satisfying Eq.~\eqref{eq:lambdaS} is
\begin{equation*}
 \bm \Lambda = \frac12\omright^{-1} \sqrt{\omright \siginf \omleft}\omleft^{-1}\, .
\end{equation*}
\end{proposition}

\begin{proof}
\emph{Explicit form of $\lam$:} The solution to Eq.~\eqref{eq:lambdaS} has to be a $\lam$ of the form
\begin{equation}
    \label{eq:change_var_lam_app}
    \lam = \frac 1 2 \siginfleft \bO \omleft^{-1} \, ,
\end{equation}
where $\siginfleft$ is any factorization of $\siginf$ of the form $\siginfleft \siginfright$ with $\siginfleft^\top = \siginfright$. One is then left with the equation $\bO \bO^\top = \id$.
Such equation is solved by \emph{any} matrix $\bO$ belonging to the orthogonal group $O(n)$, which is specified by $n(n-1)/2$ parameters. It is only by self-consistently imposing symmetry of $\lam$ that one finds the solution for $\bO$
\begin{eqnarray}
    \bO &=& \siginfleft^{-1} \omright^{-1} \sqrt{\omright \siginf \omleft } \nonumber \\
        &=& (\omright \siginfleft)^{-1} \sqrt{(\omright \siginfleft)(\omright \siginfleft)^\top}.
        \label{eq:overlap}
\end{eqnarray}
Introducing this value of $\bO$ in Eq.~\eqref{eq:change_var_lam_app} we obtain the desired result.\\
\emph{Uniqueness:} The decomposition $\om= \omleft\omright$ is not unique. Indeed, multiplying $\omleft$ by an arbitrary rotation $\bO_\Omega$ yields another possible decomposition. However the final value of $\lam$ does not depend on $\bO_\Omega$. To prove this consider the value of $\lam$ using the decomposition $\om = \omleft\bO_\Omega\bO_\Omega^\top\omright$:
\begin{align*}
 \lam &= \frac12(\bO_\Omega^\top\omright)^{-1} \sqrt{\bO_\Omega^\top\omright \siginf \omleft\bO_\Omega}(\omleft\bO_\Omega)^{-1}\\
& = \frac12\omright^{-1}\bO_\Omega \bO_\Omega^\top\sqrt{\omright \siginf \omleft}\bO_\Omega\bO_\Omega^\top\omleft^{-1}\\
& = \frac12\omright^{-1} \sqrt{\omright \siginf \omleft}\omleft^{-1}
\end{align*}
where we have used the fact that if $\sqrt{\bm Y} = \bm X $ then $\sqrt{\bO^\top\bm Y\bO} = \bO^\top\bm X\bO$. Since the argument of the square root is SPD by construction the unique value of $\lam$ is obtained by choosing the unique SPD root.

\end{proof}

\section{Estimation}
\label{app:estimation}
\subsection{The loss $\chi^2$}
First, we want to relate the loss~\eqref{eq:chisq} with the empirical mean and covariances defined in Sec.~\ref{sec:implications}.
Starting from the definition~\eqref{eq:chisq}, one can write
\begin{align*}
  \chi^2
  =& \frac 1 2 \langle (\hat \lam \bu - \Delta\bp)^\top \bm M
      (\hat \lam \bu - \Delta\bp) \rangle \, , \\
  =& \frac 1 2 \langle \bu^\top \hat \lam^\top \bm M \hat \lam \bu \rangle - \langle \bu^\top \hat \lam^\top \bm M \Delta\bp  \rangle+ \frac 1 2 \langle \Delta\bp ^\top \bm M \Delta\bp \rangle \, 
\end{align*}
Then, by plugging the definitions of the empirical covariances one is left with
\begin{equation}
 \label{eq:app_loss_explicit}
  \chi^2
  = \frac 1 2 \tr \left[ \hat \lam^\top \bm M \hat \lam \hatomd - 2\hat \lam^\top \bm M  \hatrespinfd  + \bm M\hatsig\right]
\end{equation}

\subsection{Maximum Likelihood Estimator}
\noindent\emph{Proof of Proposition \ref{prop:chi2}.} In order to derive the Maximum Likelihood estimator, we differentiate \eqref{eq:app_loss_explicit} ith respect to  $\hat\lam$:
 \begin{equation*} 
   0 = \frac{\partial \chi^2}{\partial\hat\lam}
   =\bm M (\hat\lam\hatomd - \hatrespinfd).
 \end{equation*}
It is clear that the optimal value of $\hat\lam_{\rm MLE}$ does no depend on $\bm M$. To find the explicit expression we equate the derivative to 0 and solve:
\[  \hat\lam_{\rm MLE} = \hatrespd(\hatomd)^{-1}.\]
To calculate the minimum we can plug this solution into Eq.~\eqref{eq:app_loss_explicit}, obtaining
\begin{equation}
  \chi^2_{\rm MLE} = \frac 1 2 \tr \left[ \bm M \left( \hatsig - \hatrespinfd  (\hatomd)^{-1} (\hatrespinfd)^\top\right) \right] \, .
\end{equation}
\qed

\subsection{EigenLiquidity Estimator}
The derivation of the EigenLiquidity Estimator runs along the same lines as the previous case.\\

\noindent\emph{Proof of Proposition \ref{prop:chi2ELM}.} We want to find an estimator $\hat\lam_{\rm ELM}$ that minimizes $\chi^2$ under the constraint  $[\hat\lam_{\rm ELM},\hatsig]=0$ which is equivalent to saying that  
\[\hat\lam_{\rm ELM}=\sum_{a=1}^n\hat\bs_ag_a\hat\bs_a^\top\]
where $\{\hat\bs_a\}_{a=1}^n$ are the eigenvectors of $\hatsig$. Because of this constraint, the minimization of $\chi^2$ has to be done with respect to $g_a$ rather than $\hat\lam$:
 \begin{equation*}
   0 = \frac{\partial \chi^2}{\partial g_a}
   =\tr\left(\bm M (\hat\lam\hatomd - \hatrespinfd)\hat\bs_a\hat\bs_a^\top\right)\, ,
 \end{equation*}
which for arbitrary $\bm M$ implies
\[ \hat\bs_a^\top(\hat\lam\hatomd-\hatrespinfd)\hat\bs_a = 0\]
and therefore
\[g_a = \frac{\hat\bs_a^\top\hatrespinfd\hat\bs_a}{\hat\bs_a^\top\hatomd\hat\bs_a}\,.\]

In order to compute the minimum loss we proceed as in the previous case plugging the value of $\hat\lam_{\rm ELM}$ into equation \eqref{eq:app_loss_explicit} and obtain
\begin{align*}
 \chi_{\rm ELM}^2 &= \frac12\tr\left[M\left(\hatsig + (\hat\lam_{\rm ELM}\hatomd - 2\hatrespinfd)\hat\lam_{\rm ELM}\right)\right]\\
&=\frac12\tr\left[M\left(\hatsig -\hatrespinfd\hat\lam_{\rm ELM}\right)\right]\\
& = \frac12\tr\left[M\left(\hatsig -\hatrespinfd\sum_{a=1}^n \hat\bs_a \frac{\hat\bs_a^\top\hatrespinfd\hat\bs_a}{\hat\bs_a^\top\hatomd\hat\bs_a}\hat\bs_a^\top\right)\right]\, .
\end{align*}
\qed

\subsection{Kyle Estimator}

\noindent\emph{Proof of Proposition \ref{prop:chi2Kyle}}
To prove the first statement of the proposition it suffices to compute the covariance of the efficient prices:
 \begin{align*}
   \hat\sig_{\rm Kyle}
= & \langle \hat \Delta \bp \hat \Delta \bp^\top\rangle = \hat \lam_{\rm Kyle}\langle\bu\bu^\top\rangle\hat \lam_{\rm Kyle}^\top = \hat \lam_{\rm Kyle}\hatomd\hat \lam_{\rm Kyle}^\top = k^2\hatsig\, ,
 \end{align*}
 and to use the unicity result for the above equation (proved in App.~\ref{app:explicit_linear_equil})  in order to recover Eq.~\eqref{eq:kyle_est}.\\

In order to find the loss at the minimum one has to take into account that the only free parameter is $k$, and since $\hat\lam_{\rm Kyle}\hatomd\hat\lam_{\rm Kyle} = k^2\hatsig$, the loss is 
\begin{align*}
 \chi_{\rm Kyle}^2 &= \frac12\tr\left[M\left(\hatsig + (\hat\lam_{\rm ELM}\hatomd - 2\hatrespinfd)\hat\lam_{\rm Kyle}\right)\right]\\
&=\frac12\tr\left[M\left((1 + k^2)\hatsig -2k\hatrespinfd   (\hat\omright^{\rm d})^{-1} \sqrt{\hat \omright^{\rm d} \hatsig \hat \omleft^{\rm d}}  (\hat\omleft^{\rm d})^{-1}       \right)\right]\, .
\end{align*}
Optimizing with respect to $k$ we obtain
\[\frac{\partial \chi_{\rm Kyle}^2}{\partial k} = k\tr\left[\bm M\hatsig\right] - \tr\left[\bm M \hatrespinfd  (\hat\omright^{\rm d})^{-1} \sqrt{\hat \omright^{\rm d} \hatsig \hat \omleft^{\rm d}}  (\hat\omleft^{\rm d})^{-1}   \right] = 0\ .\]
\qed\\


\subsection{Fabricating synthetic covariance matrices inspired from real data}\label{sec:fabricating}
The data set used to fabricate the synthetic covariance matrices in section \ref{sec:empirical} consists of the averaged price variations and volumes in 5 minute bins of the 2-year, 5-year, 10-year, and 30-year tenors of the U.S. Treasury Futures traded in the Chicago Board of Trade during the year 2016. 


For each pair of bonds with price variations $\Delta \bp^{(t)}$ and volumes $\bu^{(t)}$ we define the normalized covariances as
\[\bm C = \begin{pmatrix}
\hatsig & (\hatrespd)^\top\\
\hatrespinfd & \hatomd                                                                                                                                                                                                                                                                                                                                                                                                                                                                                                                                                                                                                                                                                                             \end{pmatrix}
  = \bm D^{-1}
 \begin{pmatrix}
\langle \Delta \bp\Delta\bp^\top \rangle  & \langle \bu\Delta\bp^\top \rangle \\
\langle \Delta\bp\bu^\top \rangle & \langle \bu\bu^\top \rangle
\end{pmatrix} \bm D^{-1}
 \]
where we have previously shifted the data by their empirical means and where
\[\bm D = \diag\left(\sqrt{\langle \Delta p_1^2\rangle},\sqrt{\langle \Delta p_2^2\rangle},\sqrt{\langle y_1^2\rangle}, \sqrt{\langle y_2^2\rangle}\right).\]
\paragraph{Fixing liquidity:}
Covariance matrices with the desired liquidity are constructed by rescaling the rows and columns of $\bm C$
\[\bm C_\epsilon = \begin{pmatrix}
\hatsig_\epsilon & (\hatrespd_\epsilon)^\top\\
\hatrespd_\epsilon & \hatomd_\epsilon                                                                                                                                                                                                                                                                                                                                                                                                                                                                                                                                                                                                                                                                                                             \end{pmatrix} =
\begin{pmatrix} 1&&& 0\\&1&&\\&&1&\\0&&&\sqrt{\epsilon}\end{pmatrix}\bm C \begin{pmatrix} 1&&& 0\\&1&&\\&&1&\\0&&&\sqrt{\epsilon}\end{pmatrix}.
\]
Using this recipe, the variance of the volume of the second asset will be $\epsilon$ and consequently the second column of the response is multiplied by $\sqrt{\epsilon}$. Note that the price variations are unchanged.

\paragraph{Fixing correlations:} In order to modify the price correlations we have to construct lineal combinations of assets. If $\hatsig_{12} = r$ , then
\[\bm C_\rho = \begin{pmatrix}
\hatsig_\rho & (\hatrespd_\rho)^\top\\
\hatrespd_\rho & \hatomd_\rho                                                                                                                                                                                                                                                                                                                                                                                                                                                                                                                                                                                                                                                                                                             \end{pmatrix} =
\begin{pmatrix} \bm A & 0\\0&\bm A\end{pmatrix}\bm C \begin{pmatrix} \bm A & 0\\0&\bm A\end{pmatrix}
\]
where
\[\bm A = \begin{pmatrix} 1 & 1\\1& - 1\end{pmatrix} \begin{pmatrix} \sqrt{\frac{1 + \rho}{1 + r}} & 0\\0& \sqrt{\frac{1 - \rho}{1 - r}}\end{pmatrix} \begin{pmatrix} 1 & 1\\1& - 1\end{pmatrix}.\]
This modification ensures that $\hatsig_\rho = \begin{pmatrix} 1 & \rho\\\rho&  1\end{pmatrix} $. The response and volumes are also modified accordignly. In particular the diagonal entries of $\hatomd_\rho$ will still be equal (but not necessarily equal to 1).\\

Note that for $\bm C$ SPD, both $\bm C_\epsilon$ and $\bm C_\rho$ are also SPD and therefore are coherent covariance matrices.

\end{document}